\documentclass{llncs}

\usepackage{latexsym,amssymb,amsmath,amsfonts}
\usepackage[pdftex]{graphicx}

\begin{document}

\title{Asynchronous Control-State Choreographies}

\titlerunning{Asynchronous Control-State Choreographies}

\author{Klaus-Dieter Schewe\inst{1} \and Yamine A\"{\i}t Ameur\inst{2} \and Sarah Benyagoub\inst{2}}

\authorrunning{K.-D. Schewe, Y. A\"{\i}t Ameur, S. Benyagoub}

\institute{Zhejiang University, UIUC Institute, Haining, China\\ \email{kdschewe@acm.org}
\and 
Universit\'{e} de Toulouse, IRIT/INPT-ENSEEIHT, Toulouse, France\\
\email{yamine@enseeiht.fr, sarah.benyagoub@enseeiht.fr}}

\maketitle

\begin{abstract}

Choreographies prescribe the rendez-vous synchronisation of messages in a system of communicating finite state machines. Such a system is called realisable, if the traces of the prescribed communication coincide with those of the asynchronous system of peers, where the communication channels either use FIFO queues or multiset mailboxes. In a recent article realisability was characterised by two necessary conditions that together are sufficient. A simple consequence is that realisability in the presence of a choreography becomes decidable. 
In this article we extend this work by generalising choreographies to control-state choreographies, which enable parallelism. We redefine P2P systems on grounds of control-state machines and show that a control-state choreography is equivalent to the rendez-vous compositions of its peers and that language-synchronisability coincides with synchronisability. These results are used to characterise realisability of control-state choreographies. As for the case of FSM-based choreographies we prove two necessary conditions: a sequence condition and a choice condition. Then we also show that these two conditions together are sufficient for the realisability of control-state choreographies.

\keywords{communicating system, control-state machine, asynchronous parallelism, choreography, synchronisability, peer-to-peer system, realisability}

\end{abstract}

\section{Introduction}

A peer-to-peer (P2P) system is an asynchronous system of independent peers communicating through messages. If one disregards the internal computations performed by the peers and considers only the sequences of messages sent and received, the P2P system may be seen as a system of communicating FSMs, and its semantics is defined by the traces of messages sent. In addition, stability conditions may be taken into account, i.e. only those traces are considered in which all sent messages also have been received.

Such a trace semantics can be defined in various ways, e.g. using a separate channel organised as a FIFO queue for each ordered pair of distinct peers (see e.g. \cite{brand:jacm1983,finkel:icalp2017}). In particular, messages on the same channel are received in the same order as they have been sent and no message is lost. Alternatives are the use of such FIFO queues with only a single channel for each receiver (as e.g. in \cite{basu:tcs2016}) or the organisation of the channels as multisets (see e.g. \cite{clemente:concur2014}), which corresponds to mailboxes, from which messages can be received in arbitrary order. Naturally, one may also consider the possibility of messages being lost (see e.g. \cite{chambart:concur2008}).

A common question investigated for communicating FSMs is whether the traces remain the same, if a rendez-vous (or handshake) synchronisation of (sending and receiving of) messages is considered, in which case the P2P system itself is also modelled as a FSM. This {\em synchronisability} problem has been claimed to be decidable in various publications (see e.g. \cite{basu:tcs2016}), but it was finally proven to be undecidable in general \cite{finkel:icalp2017}, though counterexamples are rather tricky.

The picture changes slightly in the presence of {\em choreographies}, i.e. FSMs that prescribe the rendez-vous synchronisation \cite{basu:popl2012}. In this case the peers are projections of a choreography, and synchronisability becomes {\em realisability} of the given choreography. As shown in \cite{schewe:foiks2020} the rendez-vous composition of the projected peers coincides with the choreography, whereas in general projections of a rendez-vous composition of arbitrary peers may not coincide with the given peers. Also the distinction between {\em language synchronisability} based only on the message traces, and {\em synchronisability} based in addition on the stable configurations reached becomes obsolete. The main result in \cite{schewe:foiks2020} shows that under these restrictive circumstances realisability can be characterised by two simple necessary conditions that both together are sufficient. Actually, a hint on the sufficiency of these conditions was already given by the compositional approach to choreographies and the associated proof of realisability \cite{benyagoub:nfm2018}. This compositional approach has been extended in \cite{benyagoub:abz2020} taking the new insights into account.

However, this final characterisation of realisable choreographies also highlights the limitations of viewing P2P systems as systems of communicating FSMs. In fact, a choreography defined by an FSM is a purely sequential system description. The purpose of a choreography is to provide a high-level prescription of how the interaction in a P2P system is supposed to proceed. Using FSMs for this is too simple to capture the needs of P2P systems even on a very high level of abstraction. For instance, a peer may send multiple messages at the same time, different peers may operate asynchronously, it may not be required that all sent messages are received, and peers may even disappear without receiving messages addressed to them. Therefore, it is necessary to provide a more sophisticated notion of choreography capturing these possible cases. It is necessary to capture the intrinsic (asynchronous) parallelism while excluding at the same time pathological behaviour that can result from arbirtrary collections of peers. Then the realisability problem can be reinvestigated on grounds of more expressive choreographies.

\paragraph{Our Contribution.}

In this article we extend our work in the conference paper \cite{schewe:foiks2020}, in which we proved necessary and sufficient conditions for the realisability of choreographies. As motivated above we replace choreographies modelled by FSMs by {\em control-state choreographies} (CSCs) modelled by control-state machines. We preserve the key idea that a choreography models the flow of messages between peers, but instead of FSMs we exploit {\em control-state machines}, which are derived from control state ASMs \cite[pp.44ff.]{boerger:2003}. The key difference is that whenever a (control) state is reached, the sending and receiving of several messages in parallel is enabled. The parallelism of sending/receiving of several messages will exploit the asynchronous semantics adopted from concurrent ASMs \cite{boerger:acin2016}.

With communicating control-state machines we can also define a generalisation of P2P systems, for which we define different notions of compositions. Using either multisets or FIFO queues as mailboxes we obtain three different types of asynchronous compositions, which we continue to call {\em peer-to-peer}, {\em mailbox} and {\em queue} composition, respectively (see \cite{schewe:foiks2020}). Using handshake communication we obtain a synchronous {\em rendez-vous} composition. In doing so we generalise the notions of {\em language synchronisability} and {\em synchronisability} to P2P systems based on control-state machines. Analogous to our previous work we show again that a CSC is equivalent to the rendez-vous composition of its projected peers. Furthermore, language synchronisability and synchronisability coincide also for P2P systems defined by CSCs. In this way our theory extension is conservative, as it preserves key features of choreography-defined P2P systems.

Finally, we investigate the generalisation of realisability for CSCs. As for the case of FSM-based choreographies we prove the necessity of two conditions: a sequence condition and a choice condition. Then we also show that these two conditions together are sufficient for the realisability of control-state choreographies. So we also obtain a generalisation of the realisability characterisation from \cite{schewe:foiks2020} to the more general control-state choreographies.

A preliminary version of this work was published in \cite{schewe:corr2020}, which is further extended here to better capture asynchronous behaviour. A shortened conference version without proofs and case study was published in \cite{schewe:medi2021}.

\paragraph{Related Work.}

The abstract view of P2P systems as communicating FSMs has already a long tradition \cite{brand:jacm1983}, and there has been a longer chain of results addressing the decidability of the (language) synchronisability problem. Decidability has been claimed by Basu et al. in \cite{basu:tcs2016} for systems with separate FIFO queues for P2P channels as well as for combined queues per receiver. For both cases Finkel and Lozes showed that (language) synchronisability is in fact undecidable \cite{finkel:icalp2017}. Assuming a mailbox semantics, i.e. multisets instead of queues, decidability can be obtained \cite{clemente:concur2014}, and this remains so even if messages can get lost \cite{chambart:concur2008}. However, the examples in \cite{finkel:icalp2017} showing that previous claims of decidability are incorrect give already a hint that if the peers are projections of their rendez-vous composition, the decability should hold. This is the case for prescribing choreographies.

These investigations apply to arbitrary systems of peers, for which an overarching FSM is composed, either using communication channels organised as queues or multisets or rendez-vous synchronisation. If the rendez-vous synchronisation is prescribed by a {\em choreography}, the picture changes, as the peers become projections of the choreography \cite{basu:popl2012}. This adds a conformity problem for choreographies \cite{basu:www2011} and extends synchronisability to the realisability of choreographies \cite{benyagoub:medi2016}. As our results in \cite{schewe:foiks2020} show, choreographies simplify the theory, as a choreography can always be regained by rendez-vous composition of its projection peers, and the problem of messages being sent but never received disappears. Furthermore, it further enables choreography repair, a problem raised in \cite{basu:fase2016} and solved in \cite{schewe:foiks2020}.

In earlier work the analogous realisability problem was studied for choreographies defined by message sequence charts. Ben-Abdallah and Leue highlight the problems that arise when a message sequence chart is underspecified \cite{benabdallah:tacas1997}, which are the same kind of problems that occur when the sequence or choice conditions are violated, while the repairs remove the underspecification. Alur et al. \cite{alur:tse2003} claim a decidability result for realisability of message sequence charts\footnote{We were not aware of their work nor did we check the correctness of the claim. Finkel and Lozes doubted the correctness of any previously claimed decidability result, but no explicit reference to the work of Alur appears in \cite{finkel:icalp2017}.}.

The characterisation of realisable choreographies in \cite{schewe:foiks2020} sets a nice end-point showing under which circumstances a high-level choreography design can be realised by an asynchronous P2P system. However, it also highlights the limitations of the notion of choreographies understood as a FSM. Parallelism is not foreseen in this approach, and this is reflected in the sequence and choice conditions in \cite[pp.273f.]{schewe:foiks2020} requiring messages to be ``non-independent''. It is not hard to imagine P2P systems that are not choreography-defined, but nonetheless specify a well-designed communication system. This motivates the investigation of more expressive notions of choreographies such as control-state choreographies that are motivated by control-state Abstract State Machines \cite[pp.44ff.]{boerger:2003}. The control-state machines used in this article are simpler as they only permit send and receive instead of arbitrary ASM rules. On the other hand they are more complex, as they integrate the theory of concurrent runs from \cite{boerger:acin2016}.

The fact that the realisability result from \cite{schewe:foiks2020} can be generalised to control-state choreographies does not imply that these are most adequate for the specification of choreographies and P2P systems. As shown in \cite{boerger:2003} by many examples control-state Abstract State Machines can be very useful on a very high level of abstraction, but they cannot replace complete state-based specifications. The extension in this article is conservative, as we preserve the focus on the message communication, whereas all other actions of the peers are disregarded. This is in line with previous work \cite{basu:tcs2016,basu:popl2012,finkel:icalp2017} and helps to clarify the realisability issue, but there are other aspects in P2P systems such as synchronisation or message content that require more information in states. Synchronisation using session types \cite{carbone:entcs2007,honda:jacm2016} allow shared agreements among the communicating peers to be integrated. Likewise, partially ordered multisets of actions \cite{guanciale:lamp2019} can be use to enhance the control of the behaviour of the peers. Bocchi et al. \cite{bocchi:lmcs2020} claim that choice should be replaced by branching controlled by conditions, which requires that such conditions must be evaluated in some state. In our view this underlines that using proper rigorous methods such as concurrent ASMs \cite{boerger:acin2016} will greatly improve the specification of P2P systems. However, for the restricted objective of this article, i.e. the generalisation of the main result in \cite{schewe:foiks2020} all these extensions are nor needed.

A constructive approach to develop realisable choreographies and consequently P2P systems was brought up in \cite{benyagoub:medi2016}. The general idea is to exploit construction operators, by means of which realisable choreographies can be built out of a primitive base. The composition operators can be specified using Event-B \cite{zoubeyr:sttt2017}, so the correctness of the construction can be verified, e.g. using the RODIN tool \cite{benyagoub:nfm2018}. This actually exploits the sufficiency of our characterisation under moderate restrictions, but cannot be used to show also necessity. On the other hand it gives already hints for choreography repair \cite{benyagoub:jsep2019}. The results in \cite{schewe:foiks2020} have been used in \cite{benyagoub:abz2020} to further strengthen the theoretical underpinnings of this correct-by-construction approach to realisable choreographies and permitted to remove unnecessary assumptions. 

Naturally, using Event-B in this context provides an open invitation to a refinement-based approach taking P2P systems defined by choreographies to communicating concurrent machines, e.g. exploiting the work on concurrent Event-B \cite{schewe:medi2018}. The proposal to support the development of concurrent systems by multiple Event-B machines with concurrent runs has been derived from concurrent ASMs \cite{boerger:acin2016}, and the introduction of messaging (as in \cite{boerger:jucs2017} for concurrent ASMs) is straightforward.

\paragraph{Organisation of the Article.}

The remainder of this article is organised as follows. Section \ref{sec:p2p} is dedicated to preliminaries, i.e. we introduce all the notions that are relevant for the work: control-state machines, P2P systems, rendez-vous, p2p, queue and mailbox semantics, and synchronisability. In Section \ref {sec:choreo} we introduce control-state choreographies and choreography-defined P2P systems, for which we show that synchronisability coincides with language-synchronisability. In Section \ref{sec:realisability} we address sufficient and necessary conditions for realisability of control-state choreographies, which gives our main result. In Section \ref{sec:ext} we briefly discuss extensions concerning infinite sets of peers and infinite sequences of messages. Finally, Section \ref{sec:fin} contains a brief summary and outlook.

\section{P2P Communication Systems and Control-State Choreographies}\label{sec:p2p}

In a P2P system we need at least peers and messages to be exchanged between them. Therefore, let $M$ and $P$ be finite, disjoint sets, elements of which are called {\em messages} and {\em peers}, respectively. Each message $m \in M$ has a unique {\em sender} $s(m) \in P$ and a unique {\em receiver} $r(m) \in P$ with $s(m) \neq r(m)$. We use the notation $i \overset{m}{\rightarrow} j$ for a message $m$ with $s(m) = i$ and $r(m) =j$. We also use the notation $!m^{i \rightarrow j}$ and $?m^{i \rightarrow j}$ for the {\em event} of sending or receiving the message $m$, respectively. Write $M_p^s$ and $M_p^r$ for the sets of messages, for which the sender or the receiver is $p$, respectively.

\subsection{Control-State Machines and P2P Systems}

Let $s(M)$ and $r(M)$ denote the sets of send and receive events defined by a set $M$ of messages. In \cite{schewe:foiks2020} we defined a P2P system over $M$ and $P$ as a family $\{ \mathcal{P}_p \}_{p \in P}$ of finite state machines\footnote{Note that the FSM $\mathcal{P}_p$ may be deterministic or non-deterministic.} (FSMs) $\mathcal{P}_p$ over an alphabet $\Sigma_p = s(M_p^s) \cup r(M_p^r)$, and by abuse of terminology $\mathcal{P}_p$ was also called a {\em peer}.

In this article we use a more general notion of P2P system based on the notion of control-state machine. This notion is derived from {\em control-state abstract state machines} \cite{boerger:2003}, but differs in that we only consider sending and receiving of messages instead of arbitrary ASM rules. 
%We also use simplified Boolean conditions, which we define by clauses: a {\em clause} is a conjunction $L_1 \wedge\dots\wedge L_k$, in which each literal $L_i$ is either $a$ or $\neg a$ with propositional atoms $a \in A$. We write $C(A)$ for the set of clauses over a finite set $A$ of propositional atoms.

Our extension is quite simple. Instead of a single transition relation, which captures choices among different send and receive events a control-state machine permits multiple transition relations. A peer can choose among these relations, while the different transition in such a relation are considered to be executed in parallel (if enabled).

\begin{definition}\label{def-ctl-machine}

A {\em control-state machine} $\mathcal{M}$ comprises 

\begin{itemize}

\item a set $Q$ of {\em control states} containing an {\em initial control state} $q_0 \in Q$,

\item an {\em alphabet} $\Sigma$, 

%\item a finite set $A$ of {\em propositional atoms}, and

%\item a finite set $\delta$ of {\em transition relations} $\tau \subseteq Q \times C(A) \times \Sigma \times C(A) \times Q$.

\item a finite set $\delta$ of {\em transition relations} $\tau \subseteq Q \times \Sigma \times Q$.

\end{itemize}

%We write $\mathcal{M} =  (Q, q_0, A, \Sigma, \delta)$.

We write $\mathcal{M} =  (Q, q_0, \Sigma, \delta)$.

\end{definition}

We use control state machines to define P2P systems. As in \cite{schewe:foiks2020} we need a set $M$ of messages and a set $P$ of associated peers, i.e. those peers $p$ that appear as sender or receiver of the messages $m \in M$ as defined above.

\begin{definition}\label{def-p2p}

A {\em peer-to-peer system} (P2P system) over a set $M$ of messages and an associated set $P$ of peers is a family $\{ \mathcal{P}_p \}_{p \in P}$ of control-state machines %$\mathcal{P}_p = (Q_p, q_{p,0}, A, \Sigma_p, \delta_p)$
$\mathcal{P}_p = (Q_p, q_{p,0}, \Sigma_p, \delta_p)$ with alphabets $\Sigma_p$ containing the sending and receiving events of messages in $M$ with sender or receiver $p$, respectively, i.e. $\Sigma_p \;=\; s(M_p^s) \;\cup\; r(M_p^r)$.

\end{definition}

By abuse of terminology we also call $\mathcal{P}_p$ a {\em peer}. %, and we often use just a finite set of indices, i.e. $P = \{ 1,\dots, n \}$.
The rationale behind Definition \ref{def-p2p} is that a P2P system is composed of several autonomous peers\footnote{In a more general context of concurrent systems in \cite{boerger:acin2016} peers are called {\em agents}. The special case, where messaging is added to shared locations is handled in \cite{boerger:jucs2017}.} that interact via messages. 
%Actually, peers also perform local actions, but for our purposes here it suffices to capture the effects of these actions by the clauses over a set of propositional atoms.

\begin{example}\label{bsp-p2p}

\ Let us consider a simple P2P system with four peers, for which we simple use numbers. We dispense with indexing messages, as different senders or receivers already indicate when we have different messages.

Peer 1 is to represent an author who sends an abstract to a program committee represented by peer 2. The program committee may either send back a confirmation or an immediate rejection. In the former case the author sends a paper to the committee and receives a decision in return. This gives rise to messages $1 \stackrel{a}{\rightarrow} 2$ (abstract), $2 \stackrel{ir}{\rightarrow} 1$ (immediate rejection), $2 \stackrel{c}{\rightarrow} 1$ (confirmation), $1 \stackrel{p}{\rightarrow} 2$ (paper), and $2 \stackrel{d}{\rightarrow} 1$ (decision).

Then for peer $\mathcal{P}_1$ we have the set of control-states $Q_1 = \{ 0,1,2,3,4,5 \}$ with initial control state $q_{1,0} = 0$, the event alphabet
\[ \Sigma_1 = \{ !a^{1 \rightarrow 2}, ?ir^{2 \rightarrow 1}, ?c^{2 \rightarrow 1}, !p^{1 \rightarrow 2}, ?d^{2 \rightarrow 1} \} \; , \]
and the set $\delta_1$ with two transition relations
\begin{align*}
\tau_{1,1} &= \{ (0,!a^{1 \rightarrow 2},1), (1,?ir^{2 \rightarrow 1},5) \} \quad\text{and} \\
\tau_{1,2} &= \{ (0,!a^{1 \rightarrow 2},1), (1,?c^{2 \rightarrow 1},2), (2,!p^{1 \rightarrow 2},3), (3,?c^{2 \rightarrow 1},4) \}
\end{align*}
The two transition relations capture the two alternative action sequences for an author peer. A submitted abstract can either lead to an immediate rejection---this is captured in $\tau_{1,1}$---or it can receive a confirmation, which then is followed by paper submission and receiving a final decision---this is captured in $\tau_{1,2}$. This control-state machine is llustrated in Figure \ref{fig-p2p_peer1}.

\begin{figure}[htb]
\begin{center}
\includegraphics[scale=0.6]{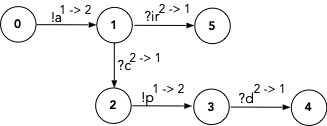}
\caption{Author Peer in P2P System\label{fig-p2p_peer1}}
\end{center}
\end{figure}

Peer 2, which represents the program committee, receives an abstract from an author (peer 1), which is either confirmed or immediately rejected. In the former case, reviewing requests are sent in parallel to two reviewers represented by peers 3 and 4. When the paper has been received from peer 1, the paper is sent for review in parallel to both reviewers. Once a review is received from peer 3 (or 4), a discussion request is sent to the other reviewer. After receiving comments from the other reviewer---this is considered to be optional---a decision is made and sent back to the author. This gives rise to additional messages $2 \stackrel{rr}{\rightarrow} 3$ and $2 \stackrel{rr}{\rightarrow} 4$ (reviewing requests), $2 \stackrel{pr}{\rightarrow} 3$ and $2 \stackrel{pr}{\rightarrow} 4$ (paper for review), $3 \stackrel{r}{\rightarrow} 2$ and $4 \stackrel{r}{\rightarrow} 2$ (reviews), $2 \stackrel{dr}{\rightarrow} 3$ and $2 \stackrel{dr}{\rightarrow} 4$ (discussion requests), and finally $3 \stackrel{c}{\rightarrow} 2$ and $4 \stackrel{c}{\rightarrow} 2$ (comments).

Peer 2 can be modelled by a control-state machine with the set $Q_2 = \{ 0, \dots, 16 \}$ of control-states, initial control state $q_{2,0} = 0$, the event alphabet
\begin{gather*}
\Sigma_2 = \{ ?a^{1 \rightarrow 2}, !ir^{2 \rightarrow 1}, !c^{2 \rightarrow 1}, ?p^{1 \rightarrow 2}, !rr^{2 \rightarrow 3}, !rr^{2 \rightarrow 4}, !pr^{2 \rightarrow 3}, !pr^{2 \rightarrow 4}, \\
?r^{3 \rightarrow 2}, ?r^{4 \rightarrow 2}, !dr^{2 \rightarrow 3}, !dr^{2 \rightarrow 4}, ?c^{3 \rightarrow 2}, ?c^{4 \rightarrow 2}, !d^{2 \rightarrow 1} \} \; ,
\end{gather*}
and the set $\delta_2$ with the following five transition relations:
\[ \tau_{2,1} = \{ (0,?a^{1 \rightarrow 2},1), (1,!ir^{2 \rightarrow 1},16) \} \]
captures again---this time considered from the angle of the programme committee peer---the case that an abstract submission by an author leads to an immediate rejection.
\begin{align*}
\tau_{2,2} &= \{ (0,?a^{1 \rightarrow 2},1), (1,!c^{2 \rightarrow 1},2), (1,!rr^{2 \rightarrow 3},3), (1,!rr^{2 \rightarrow 4},4), (2,?p^{1 \rightarrow 2},5), \\
&\qquad (5,!pr^{2 \rightarrow 3},6), (5,!pr^{2 \rightarrow 4},7), (3,?r^{3 \rightarrow 2},8), (6,?r^{3 \rightarrow 2},8), (4,?r^{4 \rightarrow 2},9), \\
&\qquad (7,?r^{4 \rightarrow 2},9), (8,!dr^{2 \rightarrow 4},10),  (9,!dr^{2 \rightarrow 3},11), (10,?c^{4 \rightarrow 2},12), (12,!d^{2 \rightarrow 1},14), \\
&\qquad (11,?c^{3 \rightarrow 2},13), (13,!d^{2 \rightarrow 1},15) \} \; , \\
\tau_{2,3} &= \{ (0,?a^{1 \rightarrow 2},1), (1,!c^{2 \rightarrow 1},2), (1,!rr^{2 \rightarrow 3},3), (1,!rr^{2 \rightarrow 4},4), (2,?p^{1 \rightarrow 2},5), \\
&\qquad (5,!pr^{2 \rightarrow 3},6), (5,!pr^{2 \rightarrow 4},7), (3,?r^{3 \rightarrow 2},8), (6,?r^{3 \rightarrow 2},8), (4,?r^{4 \rightarrow 2},9), \\
&\qquad (7,?r^{4 \rightarrow 2},9), (8,!dr^{2 \rightarrow 4},10),  (9,!dr^{2 \rightarrow 3},11), (10,?c^{4 \rightarrow 2},12), (12,!d^{2 \rightarrow 1},14), \\
&\qquad (11,!d^{2 \rightarrow 1},15) \} \; , \\
\tau_{2,4} &= \{ (0,?a^{1 \rightarrow 2},1), (1,!c^{2 \rightarrow 1},2), (1,!rr^{2 \rightarrow 3},3), (1,!rr^{2 \rightarrow 4},4), (2,?p^{1 \rightarrow 2},5), \\
&\qquad (5,!pr^{2 \rightarrow 3},6), (5,!pr^{2 \rightarrow 4},7), (3,?r^{3 \rightarrow 2},8), (6,?r^{3 \rightarrow 2},8), (4,?r^{4 \rightarrow 2},9), \\
&\qquad (7,?r^{4 \rightarrow 2},9), (8,!dr^{2 \rightarrow 4},10),  (9,!dr^{2 \rightarrow 3},11), (10,!d^{2 \rightarrow 1},14), \\
&\qquad (11,?c^{3 \rightarrow 2},13), (13,!d^{2 \rightarrow 1},15) \} \; , \\
\tau_{2,5} &= \{ (0,?a^{1 \rightarrow 2},1), (1,!c^{2 \rightarrow 1},2), (1,!rr^{2 \rightarrow 3},3), (1,!rr^{2 \rightarrow 4},4), (2,?p^{1 \rightarrow 2},5), \\
&\qquad (5,!pr^{2 \rightarrow 3},6), (5,!pr^{2 \rightarrow 4},7), (3,?r^{3 \rightarrow 2},8), (6,?r^{3 \rightarrow 2},8), (4,?r^{4 \rightarrow 2},9), \\
&\qquad (7,?r^{4 \rightarrow 2},9), (8,!dr^{2 \rightarrow 4},10),  (9,!dr^{2 \rightarrow 3},11), (10,!d^{2 \rightarrow 1},14), (11,!d^{2 \rightarrow 1},15) \} \; . \\
\end{align*}
These four transition relations have the core of the paper handling in common, i.e. the asynchronous parallel submission of a paper as well as the sending of reviewing requests to two reviewer peers, which are followed by receiving the paper submission, sending the paper to the reviewers and receiving reviews. The only alternatives occur after sending discussion requests. Comments from the other reviewer are optional, which gives rise to four choices.

This control-state machine is illustrated in Figure \ref{fig-p2p_peer2}. While the figure ressembles the presentation of FSMs by graphs it should not be confused with an FSM. An arch connecting edges stands for parallelism. In this way it becomes apparent that the transition from 3 to 8 is only possible, if in parallel the transition from 6 to 8 is executed (likely for control-states 4, 7 and 9), which means that receiving a review can only occur after the paper for review has been sent to the reviewer.

\begin{figure}[htb]
\begin{center}
\includegraphics[scale=0.6]{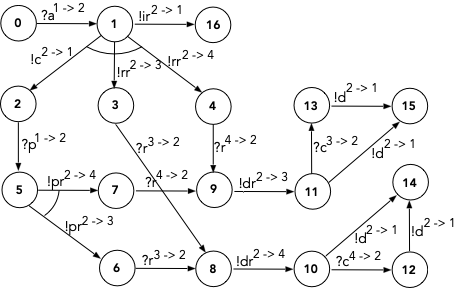}
\caption{Programme Committee Peer in P2P System\label{fig-p2p_peer2}}
\end{center}
\end{figure}

Peers 3 and 4 represent two reviewers, who first first receive a reviewing request from Peer 2 followed by the paper to be reviewed. Once a review is returned, a discussion request concerning another review is received, which may lead to comments that are sent back. These comments are considered optional. 
We can model Peer 3 by a control-state machine with control states $Q_3 = \{ 0,1,2,3,4,5 \}$, initial control-state $q_{3,0} = 0$, event alphabet
\[ \Sigma_3 = \{ ?rr^{2 \rightarrow 3}, ?pr^{2 \rightarrow 3}, !r^{3 \rightarrow 2}, ?dr^{2 \rightarrow 3}, c^{3 \rightarrow 2} \} \; , \]
and the set $\delta_3$ with two transition relations:
\begin{align*}
\tau_{3,1} &= \{ (0,?rr^{2 \rightarrow 3},1), (1,?pr^{2 \rightarrow 3},2), (2,!r^{3 \rightarrow 2},3), (3,?dr^{2 \rightarrow 3},4), (4,c^{3 \rightarrow 2},5) \} \quad\text{and}\\
\tau_{3,2} &= \{ (0,?rr^{2 \rightarrow 3},1), (1,?pr^{2 \rightarrow 3},2), (2,!r^{3 \rightarrow 2},3), (3,?dr^{2 \rightarrow 3},4) \} \; . 
\end{align*}
Peer 4 is modelled in a completely analogous way. The control-state machines are llustrated in Figure \ref{fig-p2p_peer3+4}.

\begin{figure}[htb]
\begin{center}
\includegraphics[scale=0.6]{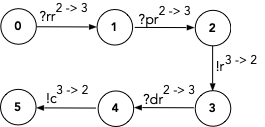} \qquad
\includegraphics[scale=0.6]{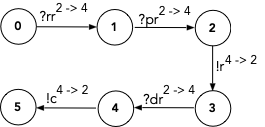}
\caption{Reviewer Peers in P2P System\label{fig-p2p_peer3+4}}
\end{center}
\end{figure}

\end{example}

For the semantics of a P2P system we have to define how the different peers interact. Informally, all peers are supposed to operate autonomously on their local states, while their interaction defines the sequences of global states. We will therefore first define the notion of state.

\begin{definition}\label{def-global-state}

A {\em state} of a P2P system $\{ \mathcal{P}_p \}_{p \in P}$ consists of

\begin{itemize}

\item a {\em combined control state} $ctl$ defined by a function $ctl: P \rightarrow \bigcup_{p \in P} Q_p$ with $ctl(p) \subseteq Q_p$,

%\item a {\em state condition}, i.e. a function $cond: A \rightarrow \{ \textbf{true}, \textbf{false} \}$ considered as a clause in $C(A)$, and

\item a {\em message pool} $B$, which is defined using one of the following four alternatives:

\begin{description}

\item[\textbf{rendez-vous.}] $B = \emptyset$;

\item[\textbf{p2p.}] $B$ is a family of FIFO queues $c_{i,j}$ indexed by pairs $(i,j)$ with $i,j \in P, i \neq j$, such that all entries in $c_{i,j}$ are messages $i \stackrel{m}{\rightarrow} j$ from sender $i$ to receiver $j$;

\item[\textbf{queue.}] $B$ is a family of FIFO queues $c_j$ indexed by peers $j \in P$, such that all entries in $c_j$ are messages $i \stackrel{m}{\rightarrow} j$ with receiver $j$;

\item[\textbf{mailbox.}] $B$ is a family of multisets $b_j$ indexed by peers $j \in P$, such that all entries in $b_j$ are messages $i \stackrel{m}{\rightarrow} j$ with receiver $j$.

\end{description}

\end{itemize}

An {\em initial state} is a state with $ctl(p) = q_{p,0}$ for all $p \in P$ and an empty message pool $B$, i.e. all queues $c_{i,j}$, $c_j$ or multisets $b_j$ for $i,j \in P$ are empty.

\end{definition}

Depending on the chosen alternative for the message pool we say that the P2P system has a p2p, queue, mailbox or rendez-vous semantics\footnote{Note that it would also be possible to use a separate mailbox for each channel defined by a pair $(i,j)$ with $i,j \in P, i \neq j$. However, as it is always possible to access all elements in a multiset, this will not make any difference. For the case of FSM-based P2P systems this was formally shown in \cite[Prop.2]{schewe:foiks2020}.}.

The rationale behind Definition \ref{def-global-state} is to understand the semantics of a P2P system as an asynchronous concurrent system that progresses by state transitions. The p2p, queue and mailbox semantics capture how systems really run; they only differ in the way messages are handled by the peers. Clearly, mailboxes provide the largest flexibility, whereas the queue semantics is the most restrictive one. However, as examples in \cite{finkel:icalp2017} show there can be P2P systems that behave oddly and hardly reflect behaviour desired from an application point of view. For this purpose we also foresee the restrictive rendez-vous semantics. With FSMs instead of control-state machines (as in \cite{finkel:icalp2017,schewe:foiks2020}) the rendez-vous semantics enforces a strictly sequential behaviour. With the use of control-state machines this behaviour is much more relaxed, and permits parallelism.

A {\em location} of a state is given by a function symbol $f$ and a $n$-tuple $\bar{a}$ of values, where $n$ is the arity of $f$. For the states defined in Definition \ref{def-global-state} the function symbols are $ctl$ of arity 1, %$cond$ of arity 1, 
and $channel$ of arity 2 in p2p semantics, $queue$ of arity 1 in queue semantics, or $mailbox$ of arity 1 in mailbox semantics---for rendez-vous semantics no further function symbol is needed. We write $val_S(\ell)$ to denote the value of the location $\ell$ in state $S$, i.e. $val_S((ctl,p)) = ctl(p)$ is defined for peers $p \in P$. %, and $val_S(cond,a) = cond(a)$ is defined for propositional atoms $a \in A$. 
For the other locations we have:

\begin{itemize}

\item In p2p semantics $val_S((channel,(i,j))) = c_{i,j}$ is defined for $i,j \in P$ with $i \neq j$.

\item In queue semantics $val_S((queue,j)) = c_j$ is defined for $j \in P$.

\item In mailbox semantics $val_S((mailbox,j)) = b_j$ is defined for $j \in P$.

\end{itemize}

For any state $S = (ctl,B)$ %$S = (ctl,cond,B)$ 
and any peer $p \in P$ we obtain a {\em projection} %$S_p = (ctl(p), cond, B_p)$
$S_p = (ctl(p), B_p)$, where $B_p$ is the family of queues $\{ c_{i,p} \mid i \in P, i \neq p \}$ in p2p semantics, the queue $c_p$ in queue semantics, the multiset $b_p$ in mailbox semantics, and $\emptyset$ in rendez-vous semantics. A projected state $S_p$ is also referred to as {\em local state} of the peer $p$.

\begin{definition}\label{def-enabled}

The set $E_p(S)$ of {\em enabled transitions} of a peer $p \in P$ in a state %$S = (ctl,cond,B)$
$S = (ctl,B)$ is defined by the following cases:

\begin{itemize}

%\item In p2p, queue or mailbox semantics a sending transition $t = (q_t, b_t, !m^{p \rightarrow j}, a_t, q_t^\prime) \in \tau \in \delta_p$ of a peer $\mathcal{P}_p$ is enabled in $S$ iff $ctl(p) = q_t$ and $cond \rightarrow b_t$ hold.

\item In p2p, queue or mailbox semantics a sending transition $t = (q_t, !m^{p \rightarrow j}, q_t^\prime) \in \tau \in \delta_p$ of a peer $\mathcal{P}_p$ is enabled in $S$ iff $q_t \in ctl(p)$ holds.

%\item In p2p semantics a receiving transition $t = (q_t, b_t, ?m^{i \rightarrow p}, a_t, q_t^\prime) \in \tau \in \delta_p$ of a peer $\mathcal{P}_p$ is enabled in $S$ iff $ctl(p) = q_t$ and $cond \rightarrow b_t$ hold, and the front element in the queue $c_{i,p}$ is the message $i \stackrel{m}{\rightarrow} p$.

\item In p2p semantics a receiving transition $t = (q_t, ?m^{i \rightarrow p}, q_t^\prime) \in \tau \in \delta_p$ of a peer $\mathcal{P}_p$ is enabled in $S$ iff $q_t \in ctl(p)$ holds, and the front element in the queue $c_{i,p}$ is the message $i \stackrel{m}{\rightarrow} p$.

%\item In queue semantics a receiving transition $t = (q_t, b_t, ?m^{i \rightarrow p}, a_t, q_t^\prime) \in \tau \in \delta_p$ of a peer $\mathcal{P}_p$ is enabled in $S$ iff $ctl(p) = q_t$ and $cond \rightarrow b_t$ hold, and the front element in the queue $c_p$ is the message $i \stackrel{m}{\rightarrow} p$.

\item In queue semantics a receiving transition $t = (q_t, ?m^{i \rightarrow p}, q_t^\prime) \in \tau \in \delta_p$ of a peer $\mathcal{P}_p$ is enabled in $S$ iff $q_t \in ctl(p)$ holds, and the front element in the queue $c_p$ is the message $i \stackrel{m}{\rightarrow} p$.

%\item In mailbox semantics a receiving transition $t = (q_t, b_t, ?m^{i \rightarrow p}, a_t, q_t^\prime) \in \tau \in \delta_p$ of a peer $\mathcal{P}_p$ is enabled in $S$ iff $ctl(p) = q_t$ and $cond \rightarrow b_t$ hold, and the message $i \stackrel{m}{\rightarrow} p$ appears in the multiset $b_p$.

\item In mailbox semantics a receiving transition $t = (q_t, ?m^{i \rightarrow p}, q_t^\prime) \in \tau \in \delta_p$ of a peer $\mathcal{P}_p$ is enabled in $S$ iff $q_t \in ctl(p)$ holds, and the message $i \stackrel{m}{\rightarrow} p$ appears in the multiset $b_p$.

%\item In rendez-vous semantics a sending transition $t_i = (q_i, b_i, !m^{i \rightarrow j}, a_i, q_i^\prime) \in \tau_i \in \delta_i$ and a receiving transition $t_j = (q_j, b_j, ?m^{i \rightarrow j}, a_j, q_j^\prime) \in \tau_j \in \delta_j$ of a peer $\mathcal{P}_j$ are enabled simultaneously in $S$ iff $ctl(i) = q_i$, $ctl(j) = q_j$ and $cond \rightarrow b_i \wedge b_j $ hold.

\item In rendez-vous semantics a sending transition $t_i = (q_i, !m^{i \rightarrow j}, q_i^\prime) \in \tau_i \in \delta_i$ and a receiving transition $t_j = (q_j, ?m^{i \rightarrow j}, q_j^\prime) \in \tau_j \in \delta_j$ of a peer $\mathcal{P}_j$ are enabled simultaneously in $S$ iff $q_i \in ctl(i)$ and $q_j \in ctl(j)$ hold.

\end{itemize}

\end{definition}

Each enabled transition $t \in E_p(S)$ yields an update set of the state $S$. In general, an {\em update} is a pair $(\ell, v)$ comprising a location $\ell$ and a value $v$, and an {\em update set} is a set of such updates. 

\begin{definition}\label{def-transition-update}

For an enabled transition $t \in E_p(S)$ of a peer $p$ in state $S$ the {\em yielded update set} $\Delta_t(S)$ is defined as follows:

\begin{itemize}

%\item If $t = (q_t, b_t, !m^{p \rightarrow j}, a_t, q_t^\prime)$ is a sending transition or $t = (q_t, b_t, ?m^{i \rightarrow p}, a_t, q_t^\prime)$ is a receiving transition, then $((ctl,p),q_t^\prime) \in \Delta_t(S)$, and for all $a \in A$ we have: if $a$ appears in $a_t$, then, $((cond,a),\textbf{true}) \in \Delta_t(S)$, and if $\neg a$ appears in $a_t$, then, $((cond,a),\textbf{false}) \in \Delta_t(S)$.

\item If $t = (q_t, !m^{p \rightarrow j}, q_t^\prime)$ is a sending transition or $t = (q_t, ?m^{i \rightarrow p}, q_t^\prime)$ is a receiving transition, then $((ctl,p),val_S(ctl,p) - \{ q_t \} \cup \{ q_t^\prime \}) \in \Delta_t(S)$.

%\item If $t = (q_t, b_t, !m^{p \rightarrow j}, a_t, q_t^\prime)$ is a sending transition, then in p2p semantics
%\[ ( (channel,(p,j)), val_S((channel,(p,j))) \oplus [p \stackrel{m}{\rightarrow} j]) \in \Delta_t(S) \; , \]
%in queue semantics
%\[ ( (queue,j), val_S((queue,j)) \oplus [p \stackrel{m}{\rightarrow} j]) \in \Delta_t(S) \; , \]
%and in mailbox semantics 
%\[ ( (mailbox,j), val_S((mailbox,j)) \uplus [p \stackrel{m}{\rightarrow} j]) \in \Delta_t(S) \; . \]

\item If $t = (q_t, !m^{p \rightarrow j}, q_t^\prime)$ is a sending transition, then in p2p semantics (using $\oplus$ to denote concatenation)
\[ ( (channel,(p,j)), val_S((channel,(p,j))) \oplus [p \stackrel{m}{\rightarrow} j]) \in \Delta_t(S) \; , \]
in queue semantics (using again $\oplus$ for concatenation)
\[ ( (queue,j), val_S((queue,j)) \oplus [p \stackrel{m}{\rightarrow} j]) \in \Delta_t(S) \; , \]
and in mailbox semantics (using $\uplus$ to denote multiset union)
\[ ( (mailbox,j), val_S((mailbox,j)) \uplus [p \stackrel{m}{\rightarrow} j]) \in \Delta_t(S) \; . \]

%\item If $t = (q_t, b_t, ?m^{i \rightarrow p}, a_t, q_t^\prime)$ is a receiving transition, then in p2p semantics
%\[ ( (channel,(i,p)), c_{i,p} ) \in \Delta_t(S) \;\text{for}\; val_S((channel,(i,p))) = [i \stackrel{m}{\rightarrow} p] \oplus c_{i,p} \; , \]
%in queue semantics $( (queue,p), c_p ) \in \Delta_t(S)$
%for $val_S((queue,p)) = [i \stackrel{m}{\rightarrow} p] \oplus c_p$, and in mailbox semantics $( (mailbox,p), val_S((mailbox,p)) - \langle i \stackrel{m}{\rightarrow} p \rangle  ) \in \Delta_t(S)$.

\item If $t = (q_t, ?m^{i \rightarrow p}, q_t^\prime)$ is a receiving transition, then in p2p semantics
\[ ( (channel,(i,p)), c_{i,p} ) \in \Delta_t(S) \;\text{for}\; val_S((channel,(i,p))) = [i \stackrel{m}{\rightarrow} p] \oplus c_{i,p} \; , \]
in queue semantics $( (queue,p), c_p ) \in \Delta_t(S)$
for $val_S((queue,p)) = [i \stackrel{m}{\rightarrow} p] \oplus c_p$, and in mailbox semantics $( (mailbox,p), val_S((mailbox,p)) - \langle i \stackrel{m}{\rightarrow} p \rangle  ) \in \Delta_t(S)$.

\item No other updates are in $\Delta_t(S)$.

\end{itemize}

\end{definition}

This notion of yielded update set generalises naturally to sets of transitions and to peers in a P2P system.

\begin{definition}\label{def-update-set}

Let $S$ be a state of a P2P system $\{ \mathcal{P}_p \}_{p \in P}$. For each set $\tau \in \delta_p$ of transitions the {\em update set yielded by $\tau$ in $S$} is $\Delta_\tau(S) = \bigcup_{t \in \tau, t \in E_p(S)} \Delta_t(S)$. The {\em set of update sets} of the peer $p \in P$ in state $S$ is $\boldsymbol{\Delta}_p(S) = \{ \Delta_\tau(S) \mid \tau \in \delta_p \}$.

\end{definition}

The definition highlights the nature of the sets of transition sets defining peers. Transitions $t$ in the same set $\tau \in \delta_p$ yield updates that are to be applied in parallel, whereas the different transition sets $\tau$ provide choices. As we build update sets by unions, we have to handle updates to the same location with function symbol $channel$, $queue$ or $mailbox$, respectively. For receiving transitions in p2p or queue semantics such updates request a particular front element in a FIFO queue, so two such updates can only occur in parallel, if they require the same front element. For receiving transitions in mailbox semantics the messages are removed according to their multiplicity. For sending transitions in p2p and queue semantics different messages are appended to the queues in arbitrary order, and they are cumulatively added to the multisets in mailbox semantics. In doing so, we actually realise {\em partial updates} as defined in \cite{schewe:ejc2011}.

Furthermore, in p2p, queue and mailbox semantics update sets yielded by different peers are handled independently from each other, whereas in rendez-vous semantics we always have to consider pairs of a sending and a receiving transition. We will take care of this dependency in the definition of runs of P2P systems.

\begin{example}\label{bsp-run}

\ We continue our previous Example \ref{bsp-p2p}. In rendez-vous semantics we have a start state $ctl_0$ with $ctl_0(1) = \{ 0 \}$, $ctl_0(2) = \{ 0 \}$, $ctl_0(3) = \{ 0 \}$ and $ctl_0(4) = \{ 0 \}$, so the sending transition $(0, !a^{1 \rightarrow 2}, 1)$ of Peer 1 and the receiving transition $(0, ?a^{1 \rightarrow 2}, 1)$ of Peer 2 are simultaneously enabled. The update yields a new state $ctl_1$ with $ctl_1(1) = \{ 1 \}$, $ctl_1(2) = \{ 1 \}$, $ctl_1(3) = \{ 0 \}$ and $ctl_1(4) = \{ 0 \}$. In this state the sending transition $(1, !ir^{2 \rightarrow 1}, 16)$ of Peer 2 and the receiving transition $(1, ?ir^{2 \rightarrow 1}, 5)$ of Peer 5 are simultaneously enabled, and the update yields a final state $ctl$ with $ctl_1(1) = \{ 5 \}$, $ctl_1(2) = \{ 16 \}$, $ctl_1(3) = \{ 0 \}$ and $ctl_1(4) = \{ 0 \}$. This run corresponds to the case, where an abstract is immediately rejected.

Alternatively, the sending transition $(1, !rr^{2 \rightarrow 3}, 3)$ of Peer 2 and the receiving transition $(0, ?rr^{2 \rightarrow 3}, 1)$ of Peer 3 are simultaneously enabled. The update yields a new state $ctl_2$ with $ctl_2(1) = \{ 1 \}$, $ctl_2(2) = \{ 1, 3 \}$, $ctl_2(3) = \{ 1 \}$ and $ctl_2(4) = \{ 0 \}$. Note that $ctl_2(2)$ contains two control-states, because only one of the parallel branches has been considered. In this state the sending transition $(1, !rr^{2 \rightarrow 4}, 4)$ of Peer 2 and the receiving transition $(0, ?rr^{2 \rightarrow 4}, 1)$ of Peer 4 are simultaneously enabled. The update yields a new state $ctl_3$ with $ctl_3(1) = \{ 1 \}$, $ctl_3(2) = \{ 1, 3, 4 \}$, $ctl_3(3) = \{ 1 \}$ and $ctl_3(4) = \{ 1 \}$. We can continue the run with the simultaneously enabled transitions $(1, !c^{2 \rightarrow 1}, 2)$ of Peer 2 and transition $(1, ?c^{2 \rightarrow 1}, 2)$ of Peer 1, which leads to a state $ctl_4$ with $ctl_4(1) = \{ 2 \}$, $ctl_4(2) = \{ 2, 3, 4 \}$, $ctl_4(3) = \{ 1 \}$ and $ctl_4(4) = \{ 1 \}$. We can further continue the run with the simultaneously enabled transitions $(2, !p^{1 \rightarrow 2}, 3)$ of Peer 1 and transition $(2, ?p^{1 \rightarrow 2}, 5)$ of Peer 2, which leads to a state $ctl_5$ with $ctl_5(1) = \{ 3 \}$, $ctl_5(2) = \{ 3, 4, 5 \}$, $ctl_5(3) = \{ 1 \}$ and $ctl_5(4) = \{ 1 \}$.

Omitting some steps we can reach a state $ctl_9$ with $ctl_9(1) = \{ 3 \}$, $ctl_9(2) = \{ 8,9 \}$, $ctl_9(3) = \{ 3 \}$ and $ctl_9(4) = \{ 3 \}$, from which we finally reach a state $ctl_{15}$ with $ctl_{15}(1) = \{ 4 \}$, $ctl_{15}(2) = \{ 14, 15 \}$, $ctl_{15}(3) = \{ 5 \}$ and $ctl_{15}(4) = \{ 5 \}$.

In p2p semantics we also start with a state with $ctl(1) = \{ 0 \}$, $ctl(2) = \{ 0 \}$, $ctl(3) = \{ 0 \}$ and $ctl(4) = \{ 0 \}$, but in addition we have channels $c_{i,j}$ for all peers $i$ and $j$, which initially are the empty list $[]$. The only enabled transition is the sending transition $(0, !a^{1 \rightarrow 2}, 1)$ of Peer 1, which yields a new state, in which $ctl(1)$ is updated to $\{ 1 \}$, and $c_{1,2}$ becomes $[1 \stackrel{a}{\rightarrow} 2 ]$. Then the only enabled transition is the receiving transition $(0, ?a^{1 \rightarrow 2}, 1)$ of Peer 2, which yields a new state, in which $ctl(2)$ is updated to $\{ 1 \}$ and $c_{1,2}$ is emptied again. In this state the transitions $(1, !rr^{2 \rightarrow 3}, 3)$, $(1, !c^{2 \rightarrow 1}, 2)$ and $(1, !rr^{2 \rightarrow 4}, 4)$ of Peer 2 are enabled, so we can reach a new state by executing any of these. We can reach a successor state with $ctl(2)$ updated to $\{ 2, 3, 4 \}$, $ctl(1)$ updated to $\{ 2 \}$, $ctl(3)$ updated to $\{ 1 \}$ and $ctl(4)$ updated to $\{ 1 \}$. The non-empty channels are $c_{2,1} = [2 \stackrel{c}{\rightarrow} 1 ]$, $c_{2,3} = [2 \stackrel{rr}{\rightarrow} 3 ]$ and $c_{2,4} = [2 \stackrel{rr}{\rightarrow} 4 ]$. 

Omitting a few steps we can reach a state with $ctl(1) = \{ 3 \}$, $ctl(2) = \{ 3, 4, 6, 7 \}$, $ctl(3) = \{ 2 \}$ and $ctl(4) = \{ 2 \}$ and again empty channels, in which the transition $(2, !r^{4 \rightarrow 2}, 3)$ of Peer 4 is enabled. We create a successor state with $ctl(1) = \{ 3 \}$, $ctl(2) = \{ 3, 4, 6, 7 \}$, $ctl(3) = \{ 2 \}$ and $ctl(4) = \{ 3 \}$, in which channel $c_{4,2}$ has been updated to $[4 \stackrel{r}{\rightarrow} 2 ]$. Then transitions $(4, ?r^{4 \rightarrow 2}, 9)$ and $(7, ?r^{4 \rightarrow 2}, 9)$ are enabled, so we create a successor state with $ctl(1) = \{ 3 \}$, $ctl(2) = \{ 3, 6, 9 \}$, $ctl(3) = \{ 2 \}$ and $ctl(4) = \{ 3 \}$ and again empty channels.

We omit further details on how to extend this run, and how to obtain alternative runs in p2p semantics. The cases of queue and mailbox semantics are analogous.

\end{example}

\subsection{Runs of P2P Systems}

We can use update sets to define successor states. For this we require consistent update sets.

\begin{definition}\label{def-consistent}

An update set $\Delta$ is {\em consistent} if and only if for any updates $(\ell,v_1)$, $(\ell, v_2) \in \Delta$ of the same location we have $v_1 = v_2$.

\end{definition}

Now assume that $S$ is a state and $\Delta$ is an arbitrary update set such the updates in $\Delta$ correspond to the semantics underlying $S$. If $\Delta$ is consistent, we define a successor state $S^\prime = S + \Delta$ by 
\[ val_{S^\prime}(\ell) \; = \; \begin{cases} v &\text{if}\; (\ell,v) \in \Delta \\
val_S(\ell) &\text{else} \end{cases} \; . \]
For inconsistent $\Delta$ we extend the definition to $S + \Delta = S$.

Using update sets of the peers yielded in states we can now express the semantics of P2P systems by the runs they permit. As peers are supposed to operate concurrently, we adopt the definition of concurrent run from \cite{boerger:acin2016}. That is, peers yielding an update set in some state contribute with these updates in building a new later state, but not necessarily the immediate successor state.

\begin{definition}\label{def-run}

A {\em run} of a P2P system $\{ \mathcal{P}_p \}_{p \in P}$ in p2p, queue or mailbox semantics is a sequence $S_0, S_1, S_2, \dots$ of states such that $S_0$ is an initial state, and $S_{i+1} = S_i + \Delta_i$ holds for all $i \ge 0$, where the update sets are defined as $\Delta_i = \bigcup_{p \in P_i} \Delta_{a(i,p)}$ with update sets $\Delta_{a(i,p)} \in \boldsymbol{\Delta}_p(S_{a(i,p)})$ formed in previous states, i.e. $a(i,p) \le i$, for subsets $P_i \subseteq P$ of peers such that $p \in P_i$ implies $p \notin P_j$ for all $a(i,p) \le j < i$.

\end{definition}

According to our remark at the end of the previous subsection Definition \ref{def-run} is sufficient to define the runs for P2P systems with p2p, queue or mailbox semantics. For rendez-vous semantics the update sets $\Delta_i$ in the definition require further restrictions.

\begin{definition}\label{def-rv-run}

A {\em run} of a P2P system $\{ \mathcal{P}_p \}_{p \in P}$ in rendez-vous semantics is a sequence $S_0, S_1, S_2, \dots$ of states as in Definition \ref{def-run}, where in addition the update sets $\Delta_{a(i,p)}$ satisfy the following condition: if $\Delta_{a(i,p)}$ contains an update defined by a sending transition $t$ with the sending event $!m^{p \rightarrow j}$, then $a(i,j) = a(i,p)$ and $\Delta_{a(i,j)} \in \boldsymbol{\Delta}_j(S_{a(i,j)})$ contains the update defined by the corresponding receiving transition $t^\prime$ with the receiving event $?m^{p \rightarrow j}$ and vice versa.

\end{definition}

Extending our previous work on the basis FSMs we can now define message languages for our generalised P2P systems. For a state transition from $S_i$ to $S_{i+1}$ in a run consider the multiset $M_i$ containing those messages $m \in M$ that appear in an update in $\Delta_i$ defined by a sending transition. If $m$ appears in several such updates, it appears in $M_i$ with the corresponding multiplicity. Then let $\hat{M}_i$ be the set of ordered sequences containing all the elements of $M_i$.

\begin{definition}

For a run $R= S_0, S_1, \dots$ let $\mathcal{L}(R)$ be the set of all sequences with elements in $M$ that result from concatenation of $\hat{M}_0, \hat{M}_1, \dots$. Let $\mathcal{L}$ be the language of all finite sequences in $\bigcup_{R} \mathcal{L}(R)$, where the union is built over all runs of the P2P systems. We write $\mathcal{L}_0$, $\mathcal{L}_{p2p}$, $\mathcal{L}_q$ and $\mathcal{L}_m$ for the languages in rendez-vous, p2p, queue and mailbox semantics, respectively, and these languages the {\em trace languages} of the P2P system in rendez-vous, p2p, queue and mailbox semantics, respectively.

Two P2P systems are called {\em equivalent }in rendez-vous, p2p, queue and mailbox semantics, respectively, if and only if they define the same trace language.

\end{definition}

Note that finite sequences correspond to finite runs. So we can also define sublanguages $\hat{\mathcal{L}}_{p2p}$, $\hat{\mathcal{L}}_q$ and $\hat{\mathcal{L}}_m$ containing only those sequences of message that result from runs $S_0, S_1, \dots, S_k$, where all channels or mailboxes in the state $S_k$ are empty. These languages are called the {\em stable trace languages} of the P2P system in the different semantics.

We use these generalised trace and stable trace languages associated with a P2P system to generalise the notions of synchronisability and language-synchronis\-ability from \cite{schewe:foiks2020} and \cite{finkel:icalp2017}.

\begin{definition}\label{def-sync}

A P2P system $\{ \mathcal{P}_p \}_{p \in P}$ is {\em language-synchronisable} in p2p, queue or mailbox semantics, respectively, if and only if $\mathcal{L}_0 = \mathcal{L}_{p2p}$, $\mathcal{L}_0 = \mathcal{L}_q$ or $\mathcal{L}_0 = \mathcal{L}_m$ holds, respectively.

A P2P system $\{ \mathcal{P}_p \}_{p \in P}$ is {\em synchronisable} in p2p, queue or mailbox semantics, respectively, if and only if in addition $\hat{\mathcal{L}}_{p2p} = \mathcal{L}_{p2p}$, $\hat{\mathcal{L}}_q = \mathcal{L}_q$ or $\hat{\mathcal{L}}_m = \mathcal{L}_m$ holds, respectively.

\end{definition}

Following our previous remarks about the rationale behind P2P systems defined as families of control-state machines the notion of language-synchronisability emphasises a desirable property of P2P systems. On an abstract level we may view the system via its rendez-vous semantics, in which messaging is handled in an atomic way, whereas in p2p, queue and mailbox semantics we emphasise the reification by autonomous peers with the only difference concerning how messages are kept and processed. The stricter notion of synchronisability further emphasises the desire that messages should not be allowed to be ignored by the receiver peer.

Clearly, synchronisability implies language-synchronisability. These two properties are known to be undecidable in general even for the restricted case, where the peers are defined by FSMs \cite{finkel:icalp2017}.

\begin{example}\label{bsp-traces}

\ Following up on Example \ref{bsp-run} we can determine the trace and stable trace languages. We see that $\mathcal{L}_0 = \mathcal{L}_{p2p} = \mathcal{L}_q = \mathcal{L}_m = \mathcal{L}$ hold, and also $\hat{\mathcal{L}}_{p2p} = \mathcal{L}_{p2p}$, $\hat{\mathcal{L}}_q = \mathcal{L}_q$ and $\hat{\mathcal{L}}_m = \mathcal{L}_m$ hold, where $\mathcal{L}$ is the language containing the sequences $a^{1 \rightarrow 2} ir^{2 \rightarrow 1}$ and all $a^{1 \rightarrow 2} w_1 w_2 d^{2 \rightarrow 1}$, where $w_1$ is any interleaving of $c^{2 \rightarrow 1} p^{1 \rightarrow 2}$, $rr^{2 \rightarrow 3}$ and $rr^{2 \rightarrow 4}$, and $w_2$ is any interleaving of $pr^{2 \rightarrow 3} r^{3 \rightarrow 2} dr^{2 \rightarrow 4} c^{4 \rightarrow 2}$ with optional $c^{4 \rightarrow 2}$, and $pr^{2 \rightarrow 4} r^{4 \rightarrow 2} dr^{2 \rightarrow 3} c^{3 \rightarrow 2}$ with optional $c^{3 \rightarrow 2}$.

In particular, the P2P system from Example \ref{bsp-p2p} is synchronisable and language-synchronisable in p2p, queue and mailbox semantics.

\end{example}

\section{Choreography-Defined P2P Systems}\label{sec:choreo}

While synchronisability and language-synchronisability are undecidable in general, it was shown in \cite{schewe:foiks2020} that decidability results, if we assume P2P systems that are choreography-defined.

\subsection{Control-State Choreographies} 

We first define a generalised notion of choreography based on control-state machines, which we call control-state choreographies.

\begin{definition}\label{def-choreography}

Let $M$ be a set of messages with peers $P$. A {\em control-state choreography} (CSC) over $M$ and $P$ is a control-state machine $\mathcal{C} = (Q, q_0, M, \delta)$. %$\mathcal{C} = (Q, q_0, A, M, \delta)$.

\end{definition}

We define the semantics of CSCs analogous to the rendez-vous semantics of a P2P system with a single peer. The fact that we now only consider the messages instead of the associated sending and receiving events eases the formulation of runs.

\begin{definition}\label{def-state}

%A {\em state} of a CSC $(Q, q_0, A, M, \delta)$ consists of a control state $ct \in Q$ and a state condition defined by a function $cond: A \rightarrow \{ \textbf{true}, \textbf{false} \}$ (considered as a clause in $C(A)$). An {\em initial state} is a state with $ct = q_0$.

A {\em state} of a CSC $\mathcal{C} = (Q, q_0, M, \delta)$ is a set of control states $ct \subseteq Q$. An {\em initial state} is $ct = \{ q_0 \}$.

%A transition $t = (q, \phi, i \stackrel{m}{\rightarrow} j, \psi, q^\prime) \in \tau \in \delta$ is {\em enabled} in state $S = (ct,cond)$ iff $ct = q$ and $cond \rightarrow \phi$ hold. Then the {\em update set $\Delta_t(S)$} contains exactly the updates $(ct, q^\prime))$, $((cond,a),\textbf{true})$ for all $a$ appearing in $\psi$, and $((cond,a),\textbf{false})$ for all $\neg a$ appearing in $\psi$.

A transition $t = (q, i \stackrel{m}{\rightarrow} j, q^\prime) \in \tau \in \delta$ is {\em enabled} in state $S= ct$ iff $q \in ct$ holds; we write $t \in E_{\mathcal{C}}(S)$. Then the {\em update set $\Delta_t(S)$} contains exactly the update $(ct, ct - \{ q \} \cup \{ q^\prime \}))$.

For each set $\tau \in \delta$ the {\em update set yielded by $\tau$ in $S$} is $\Delta_\tau(S) = \bigcup_{t \in \tau, t \in E_{\mathcal{C}}(S)} \Delta_t(S)$, and the {\em set of update sets} in state $S$ is $\boldsymbol{\Delta}_{\mathcal{C}}(S) = \{ \Delta_\tau(S) \mid \tau \in \delta \}$.

\end{definition}

All remarks in the preceding section we made in connection with the partial updates apply analogously to the update sets defined by CSCs. We now obtain the notion of run of a CSC.

\begin{definition}\label{def-csc-run}

A {\em run} of a CSC $\mathcal{C} = (Q, q_0, M, \delta)$ is a sequence $S_0, S_1, S_2, \dots$ of states such that $S_0$ is an initial state, and $S_{i+1} = S_i + \Delta_i$ holds for an update set $\Delta_i \in \boldsymbol{\Delta}_{\mathcal{C}}(S_i)$.

\end{definition}

For a state transition from $S_i$ to $S_{i+1}$ in a run consider the multiset $M_i$ containing those messages $m \in M$ that appear in the update in $\Delta_i$. If $m$ appears in several such updates, it appears in $M_i$ with the corresponding multiplicity. Then let $\hat{M}_i$ be the set of ordered sequences containing all the elements of $M_i$.

For a run $R= S_0, S_1, \dots$ let $\mathcal{L}(R)$ be the set of all sequences with elements in $M$ that result from concatenation of $\hat{M}_0, \hat{M}_1, \dots$. We define the {\em trace language} $\mathcal{L}(\mathcal{C})$ of the CSC as the language of all finite sequences in $\bigcup_{R} \mathcal{L}(R)$, where the union is built over all runs $R$ of the CSC. 

\begin{example}\label{bsp-choreo}

\ We define a choreography by a control-state machine with the set $Q = \{ 0, \dots, 15 \}$ of control-states, initial control state $q_0 = 0$, the messages
\begin{gather*}
M = \{ 1 \stackrel{a}{\rightarrow} 2, 2 \stackrel{ir}{\rightarrow} 1, 2 \stackrel{c}{\rightarrow} 1, 1 \stackrel{p}{\rightarrow} 2, 2 \stackrel{rr}{\rightarrow} 3, 2 \stackrel{rr}{\rightarrow} 4, 2 \stackrel{pr}{\rightarrow} 3, 2 \stackrel{pr}{\rightarrow} 4, 3 \stackrel{r}{\rightarrow} 2, \\
4 \stackrel{r}{\rightarrow} 2, 2 \stackrel{dr}{\rightarrow} 3, 2 \stackrel{dr}{\rightarrow} 4, 3 \stackrel{c}{\rightarrow} 2, 4 \stackrel{c}{\rightarrow} 2, 2 \stackrel{d}{\rightarrow} 1 \} \; ,
\end{gather*}
and the set $\delta$ with the following five transition relations:
\begin{align*}
\tau_1 &= \{ (0, 1 \stackrel{a}{\rightarrow} 2,1), (1,2 \stackrel{ir}{\rightarrow} 1,15) \} \; , \\
\tau_2 &= \{ (0, 1 \stackrel{a}{\rightarrow} 2, (1, 2 \stackrel{c}{\rightarrow} 1,2), (1,2 \stackrel{rr}{\rightarrow} 3,3), (1,2 \stackrel{rr}{\rightarrow} 4,4), (2,1 \stackrel{p}{\rightarrow} 2,5), \\
&\qquad (5,2 \stackrel{pr}{\rightarrow} 3,6), (5,2 \stackrel{rr}{\rightarrow} 4,7), (3,3 \stackrel{r}{\rightarrow} 2,8), (6,3 \stackrel{r}{\rightarrow} 2,8), (4,4 \stackrel{r}{\rightarrow} 2,9), \\
&\qquad (7,4 \stackrel{r}{\rightarrow} 2,9), (8,2 \stackrel{dr}{\rightarrow} 3,10),  (9,2 \stackrel{dr}{\rightarrow} 4,11), (10,4 \stackrel{c}{\rightarrow} 2,12), (12,2 \stackrel{d}{\rightarrow} 1,14), \\
&\qquad (11,3 \stackrel{c}{\rightarrow} 2,13), (13,2 \stackrel{d}{\rightarrow} 1,14) \} \; , \\
\tau_3 &= \{ (0, 1 \stackrel{a}{\rightarrow} 2, (1, 2 \stackrel{c}{\rightarrow} 1,2), (1,2 \stackrel{rr}{\rightarrow} 3,3), (1,2 \stackrel{rr}{\rightarrow} 4,4), (2,1 \stackrel{p}{\rightarrow} 2,5), \\
&\qquad (5,2 \stackrel{pr}{\rightarrow} 3,6), (5,2 \stackrel{rr}{\rightarrow} 4,7), (3,3 \stackrel{r}{\rightarrow} 2,8), (6,3 \stackrel{r}{\rightarrow} 2,8), (4,4 \stackrel{r}{\rightarrow} 2,9), \\
&\qquad (7,4 \stackrel{r}{\rightarrow} 2,9), (8,2 \stackrel{dr}{\rightarrow} 3,10),  (9,2 \stackrel{dr}{\rightarrow} 4,11), (10,2 \stackrel{d}{\rightarrow} 1,14), \\
&\qquad (11,3 \stackrel{c}{\rightarrow} 2,13), (13,2 \stackrel{d}{\rightarrow} 1,14) \} \; , \\
\tau_4 &= \{ (0, 1 \stackrel{a}{\rightarrow} 2, (1, 2 \stackrel{c}{\rightarrow} 1,2), (1,2 \stackrel{rr}{\rightarrow} 3,3), (1,2 \stackrel{rr}{\rightarrow} 4,4), (2,1 \stackrel{p}{\rightarrow} 2,5), \\
&\qquad (5,2 \stackrel{pr}{\rightarrow} 3,6), (5,2 \stackrel{rr}{\rightarrow} 4,7), (3,3 \stackrel{r}{\rightarrow} 2,8), (6,3 \stackrel{r}{\rightarrow} 2,8), (4,4 \stackrel{r}{\rightarrow} 2,9), \\
&\qquad (7,4 \stackrel{r}{\rightarrow} 2,9), (8,2 \stackrel{dr}{\rightarrow} 3,10),  (9,2 \stackrel{dr}{\rightarrow} 4,11), (10,4 \stackrel{c}{\rightarrow} 2,12), (12,2 \stackrel{d}{\rightarrow} 1,14), \\
&\qquad (11,2 \stackrel{d}{\rightarrow} 1,14) \} \; , \\
\tau_5 &= \{ (0, 1 \stackrel{a}{\rightarrow} 2, (1, 2 \stackrel{c}{\rightarrow} 1,2), (1,2 \stackrel{rr}{\rightarrow} 3,3), (1,2 \stackrel{rr}{\rightarrow} 4,4), (2,1 \stackrel{p}{\rightarrow} 2,5), \\
&\qquad (5,2 \stackrel{pr}{\rightarrow} 3,6), (5,2 \stackrel{rr}{\rightarrow} 4,7), (3,3 \stackrel{r}{\rightarrow} 2,8), (6,3 \stackrel{r}{\rightarrow} 2,8), (4,4 \stackrel{r}{\rightarrow} 2,9), \\
&\qquad (7,4 \stackrel{r}{\rightarrow} 2,9), (8,2 \stackrel{dr}{\rightarrow} 3,10),  (9,2 \stackrel{dr}{\rightarrow} 4,11), (10,2 \stackrel{d}{\rightarrow} 1,14), (11,2 \stackrel{d}{\rightarrow} 1,14) \} \; .
\end{align*}

Note that final states 14 and 15 could be merged without changing the semantics. The choreography is illustrated in Figure \ref{fig-choreo}. The projected P2P system is the one in Example \ref{bsp-p2p}. Note that $3 \stackrel{r}{\rightarrow} 2$ (likewise $4 \stackrel{r}{\rightarrow} 2$) are only enabled, if both control-states 3 and 6 (4 and 7, respectively) have been reached.

\end{example}

\begin{figure}[htb]
\begin{center}
\includegraphics[scale=0.6]{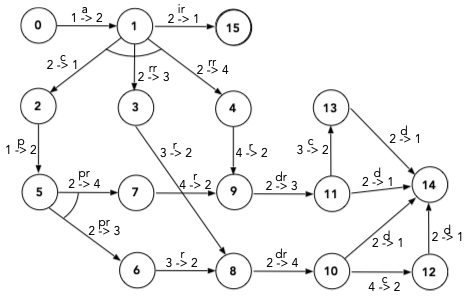}
\caption{Realisable Choreography\label{fig-choreo}}
\end{center}
\end{figure}

\subsection{Choreographies Defined by P2P Systems}

Now consider a P2P system $\{ \mathcal{P}_p \}_{p \in P}$. For each peer $p \in P$ and each transition set $\tau \in \delta_p$ we may assume without loss of generality that transitions $t, t^\prime \in \tau$ have the same initial and final control states $q_{\tau}$ and $q_{\tau}^\prime$, respectively. Otherwise, the transition set $\tau$ could be split without changing the semantics. According to Definition \ref{def-enabled} equal initial control states are necessary for the transitions to be enabled simultaneously, and according to Definition \ref{def-transition-update} equal final control states are necessary to ensure that consistent update sets are yielded. With this assumption we can select a transition set $\tau_p \in \delta_p$ for each peer $p \in P$ are choose $\tau_p = \emptyset$. 

\begin{definition}\label{def-pairing} 

A selection $\{ \tau_p \mid p \in P \}$ such that either $\tau_p \in \delta_p$ or $\tau_p = \emptyset$ holds  is {\em closed under pairing} iff whenever $\tau_p$ contains a sending transition %$t = (q_{\tau_p}, \phi, !m^{p \rightarrow p^\prime}, \psi, q_{\tau_p}^\prime)$, then $\tau_{p^\prime}$ contains the corresponding receiving transition $t^\prime = (q_{\tau_{p^\prime}}, \phi, ?m^{p \rightarrow p^\prime}, \psi, q_{\tau_{p^\prime}}^\prime)$ and vice versa.
$t = (q_{\tau_p}, !m^{p \rightarrow p^\prime}, q_{\tau_p}^\prime)$, then $\tau_{p^\prime}$ contains the corresponding receiving transition $t^\prime = (q_{\tau_{p^\prime}}, ?m^{p \rightarrow p^\prime}, q_{\tau_{p^\prime}}^\prime)$ and vice versa.

\end{definition}

If a selection $\{ \tau_p \mid p \in P \}$ is closed under pairing, we can define tuples $q, q^\prime: P \rightarrow \bigcup_{p \in P} Q_p$ with $q(p) = q_{\tau_p} \in Q_p$ and $q^\prime(p) = q_{\tau_p}^\prime \in Q_p$ for all $p \in P$; in case $\tau_p = \emptyset$ we extend $q$ and $q^\prime$ chosing $q(p) = q^\prime(p) \in Q_p$ arbitrarily. Then the set
%\[ \bigcup_{p, p^\prime \in P} \{ (q, \phi, p \stackrel{m}{\rightarrow} p^\prime, \psi, q^\prime) \mid (q_{\tau_p}, \phi, !m^{p \rightarrow p^\prime}, \psi, q_{\tau_p}^\prime) \in \tau_p, (q_{\tau_{p^\prime}}, \phi, ?m^{p \rightarrow p^\prime}, \psi, q_{\tau_{p^\prime}}^\prime) \in \tau_{p^\prime} \} \]
\[ \bigcup_{p, p^\prime \in P} \{ (q, p \stackrel{m}{\rightarrow} p^\prime, q^\prime) \mid (q_{\tau_p}, !m^{p \rightarrow p^\prime}, q_{\tau_p}^\prime) \in \tau_p, (q_{\tau_{p^\prime}}, ?m^{p \rightarrow p^\prime},  q_{\tau_{p^\prime}}^\prime) \in \tau_{p^\prime} \} \]
is the {\em set of choreography transitions} of the selection.

\begin{definition}\label{def-p2p-csc}

The {\em choreography defined by a P2P system} $\{ \mathcal{P}_p \}_{p \in P}$ is %$\mathcal{C} = (Q, q_0, A, M, \delta)$, 
$\mathcal{C} = (Q, q_0, M, \delta)$, where the set $Q$ of control states is the set of all tuples $q: P \rightarrow \bigcup_{p \in P} Q_p$ with $q(p) \in Q_p$, the initial control state is $q_0$ with $q_0(p) = q_{p,0}$ for all $p \in P$ and $\delta$ contains exactly the non-empty sets of choreography transitions defined by all selections that are closed under pairing.

\end{definition}

\begin{example}

\ The choreography defined by the P2P system in Example \ref{bsp-p2p} is the choreography in Example \ref{bsp-choreo}.

\end{example}

\begin{proposition}\label{prop-p2p-csc}

The trace language $\mathcal{L}(\mathcal{C})$ of the choreography $\mathcal{C}$ defined by a P2P system $\{ \mathcal{P}_p \}_{p \in P}$ is equal to the trace language $\mathcal{L}_0$ defined by the rendez-vous semantics of the P2P system.

\end{proposition}

\begin{proof}

We first show $\mathcal{L}(\mathcal{C}) \subseteq \mathcal{L}_0(\{ \mathcal{P}_p \}_{p \in P})$. Whenever we have $m_0 \dots m_\ell \in \mathcal{L}(\mathcal{C})$, there exists a run $S_0, S_1, \dots, S_{k+1}$ of $\mathcal{C}$ with $\hat{m}_i = m_{f(i)} \dots, m_{\ell(i)} \in \hat{M}_i$ for $i = 0, \dots, k$---$f(i) \le \ell(i)$, $f(i+1) = \ell(i) + 1$, $f(0) = 0$ and $\ell(k) = \ell$. In this run we have $S_{i+1} = S_i + \Delta_{\tau_i}(S_i)$ for $i = 0, \dots, k$ with $\Delta_{\tau_i}(S_i) = \bigcup_{t \in \tau_i, t \in E_{\mathcal{C}}(S_i)} \Delta_t(S_i) = \{ (ct, q_{i+1}) \}$. 

The states $q_i$ ($i = 0, \dots, k$) are defined as tuples $P \rightarrow \bigcup_{p \in P} Q_p$ with $q_i(p) \in Q_p$ and in particular $q_0(p) = q_{p,0}$, and the transition sets are
\[ \tau_i = \bigcup_{p, p^\prime \in P} \{ (q, p \stackrel{m}{\rightarrow} p^\prime, q^\prime) \mid (q_{\tau_p}, !m^{p \rightarrow p^\prime}, q_{\tau_p}^\prime) \in \tau_p, (q_{\tau_{p^\prime}}, ?m^{p \rightarrow p^\prime},  q_{\tau_{p^\prime}}^\prime) \in \tau_{p^\prime} \} \]
defined by a selection $\{ \tau_p \mid p \in P \}$ that is closed under pairing, i.e. we always have $(q_i(p), !m^{p \rightarrow p^\prime}, q_{i+1}(p)) \in \tau_p$ and $(q_i(p^\prime), ?m^{p \rightarrow p^\prime}, q_{i+1}(p^\prime)) \in \tau_{p^\prime}$ with messages $m \in M_i$.

Then we can choose $P_i = \{ p \in P \mid \tau_p \neq \emptyset \}$ and $a(i,p) = i$. This defines a run $\bar{S}_0, \bar{S}_1, \dots, \bar{S}_{k+1}$ of $\{ \mathcal{P}_p \}_{p \in P}$ in rendez-vous semantics with update sets $\bar{\Delta}_{a(i,p)} \in \boldsymbol{\Delta}_p(\bar{S}_i)$ such that we have $\bar{\Delta}_{a(i,p)} = \bigcup_{t \in \tau_p, t \in E_p(S_i)} \bar{\Delta}_t(S_i)$. Due to the closure under pairing the multisets of messages processed by these update sets are exactly the multisets $M_i$, the corresponding sublanguages are $\hat{M}_i$ for $i = 0, \dots, k$, and hence $m_0 \dots m_\ell \in \mathcal{L}_0(\{ \mathcal{P}_p \}_{p \in P})$.

Conversely, we have to show $\mathcal{L}_0(\{ \mathcal{P}_p \}_{p \in P}) \subseteq \mathcal{L}(\mathcal{C})$. Whenever we now have $m_0 \dots m_\ell$ $\in \mathcal{L}_0(\{ \mathcal{P}_p \}_{p \in P})$, there exists a run $\bar{S}_0, \bar{S}_1, \dots, \bar{S}_{k+1}$ of $\{ \mathcal{P}_p \}_{p \in P}$ in rendez-vous semantics with $\hat{m}_i \in \hat{M}_i$ for $i = 0, \dots, k$ and the same conditions as above. Then we have $\bar{S}_{i+1} = \bar{S}_i + \bar{\Delta}_{\tau_i}$ for $i = 0, \dots, k$ with update sets $\bar{\Delta}_{\tau_i} = \bigcup_{p \in P_i} \bar{\Delta}_{a(i,p)}$, where $a(i,p) \le i$ and $\bar{\Delta}_{a(i,p)} \in \boldsymbol{\Delta}_p(\bar{S}_{a(i,p)})$ hold.

However, any update set in $\boldsymbol{\Delta}_p(\bar{S}_j)$ contains only updates of locations $(q,p)$, Therefore, we also have $\bar{\Delta}_{a(i,p)} \in \boldsymbol{\Delta}_p(\bar{S}_i)$ or equivalently $a(i,p) = i$. The additional condition in Definition \ref{def-rv-run} allows us to define a transition $\tau_i$ by means of a selection $\{ \tau_p \mid p \in P \}$ that is closed under pairing with $q_{\tau_p} = q_i(p)$, $q^\prime_{\tau_p} = q_{i+1}(p)$ and $(q_{\tau_p}, !m^{p \rightarrow p^\prime}, q_{\tau_p}^\prime) \in \tau_p \wedge (q_{\tau_{p^\prime}}, ?m^{p \rightarrow p^\prime},  q_{\tau_{p^\prime}}^\prime) \in \tau_{p^\prime}$ if and only if $(q_i, p \stackrel{m}{\rightarrow} p^\prime, q_{i+1}) \in \Delta_t(\bar{S}_i)$ for an enabled transition $t \in \tau \in \delta$. Then the same sets $M_i$ result from a run $S_0, S_1, \dots, S_{k+1}$ of $\mathcal{C}$ with states $q_i$ defined by $q_i(p) = val_{\bar{S}_i}(ctl,p)$, which implies $m_0 \dots m_\ell \in \mathcal{L}(\mathcal{C})$.\qed

\end{proof}

\subsection{P2P Systems Defined by Choreographies}

We now investigate the inverse of Proposition \ref{prop-p2p-csc}, i.e. we will define a P2P system associated with a given choreography and show that the trace language $\mathcal{L}_0$ defined by its rendez-vous semantics is the same as the trace language of the choreography, which generalises \cite[Prop.4]{schewe:foiks2020}. For this we will introduce projections to peers, which will be straightforward except for a little subtlety associated with $\epsilon$-transitions, in which a control state is changed without a sending or receiving event. Analogous to \cite[Prop.2]{schewe:foiks2020} such $\epsilon$-transitions can be eliminated; they are neither needed nor particularly useful.

%So let $\mathcal{C} = (Q, q_0, A, M, \delta)$ be a CSC over $M$ and $P$. 
So let $\mathcal{C} = (Q, q_0, M, \delta)$ be a CSC over $M$ and $P$. For each $p \in P$ we define the {\em projection} $\mathcal{C}_p$ of $\mathcal{C}$ in two steps:

\begin{description}

%\item[\textbf{Step 1.}] First define $\mathcal{C}_p^\prime = (Q, q_0, A, \Sigma_p, \delta_p^\prime)$ with alphabet $\Sigma_p = s(M_p^s) \cup r(M_p^r)$ and set of transition sets $\delta_p^\prime = \{ \tau_p \mid \tau \in \delta \}$, where
%\[ \tau_p = \{ (q,\phi,\pi_p(m),\psi,q^\prime) \mid (q,\phi,m,\psi,q^\prime) \in \tau \} \]
%with
%\[ \pi_p(i \stackrel{m}{\rightarrow} j) = \begin{cases} !m^{p \rightarrow j} &\text{if}\; i = p \\
%?m^{i \rightarrow p} &\text{if}\; j = p \\ \epsilon &\text{else} \end{cases} \; . \]

\item[\textbf{Step 1.}] First define $\mathcal{C}_p^\prime = (Q, q_0, \Sigma_p, \delta_p^\prime)$ with alphabet $\Sigma_p = s(M_p^s) \cup r(M_p^r)$ and set of transition sets $\delta_p^\prime = \{ \tau_p \mid \tau \in \delta \}$, where $\tau_p = \{ (q,\pi_p(m),q^\prime) \mid (q,m,q^\prime) \in \tau \}$ with
\[ \pi_p(i \stackrel{m}{\rightarrow} j) = \begin{cases} !m^{p \rightarrow j} &\text{if}\; i = p \\
?m^{i \rightarrow p} &\text{if}\; j = p \\ \epsilon &\text{else} \end{cases} \; . \]

Here, $\mathcal{C}_p^\prime$ is not a control-state machine as defined in Definition \ref{def-ctl-machine}, because it admits transitions, in which only the control state is modified, but no message is sent nor received. Nonetheless, for such {\em $\epsilon$-extended machines} we can still define enabled transitions as in Definition \ref{def-enabled} and update sets as in Definition \ref{def-update-set}---the only difference concerns Definition \ref{def-transition-update} in that $\epsilon$-transitions do not yield updates to locations with function symbols $channel$, $queue$ or $mailbox$ in p2p, queue and mailbox semantics, respectively. Therefore, we can treat $\{ \mathcal{C}_p^\prime \}_{p \in P}$ also as a P2P system.

%\item[\textbf{Step 2.}] We can eliminate the $\epsilon$-transitions and define a control state machine $\mathcal{C}_p = (Q, q_0, A, \Sigma_p, \delta_p)$. As before we can assume without loss of generality that for each $\tau \in \delta_p^\prime$ we have unique states $q_\tau$ and $q_\tau^\prime$ such that the transitions $t \in \tau$ all have the form $(q_\tau, \phi_i, \tilde{m}_i, \psi_i, q_\tau^\prime)$; otherwise we could split $\tau$. Furthermore, we can modify the clauses such that without loss of generality the clauses $\phi_i$ are either equal or incompatible, i.e. $\phi_i$ contains some literal $L$ and $\phi_j$ contains its negation.

\item[\textbf{Step 2.}] We can eliminate the $\epsilon$-transitions and define a control state machine $\mathcal{C}_p = (Q, q_0, \Sigma_p, \delta_p)$. As before we can assume without loss of generality that for each $\tau \in \delta_p^\prime$ we have unique states $q_\tau$ and $q_\tau^\prime$ such that the transitions $t \in \tau$ all have the form $(q_\tau, \tilde{m}_i, q_\tau^\prime)$; otherwise we could split $\tau$. 

%Then take transitions $t = (q,\phi,\epsilon,\psi,q^\prime) \in \tau$ and $t^\prime = (q^\prime, \phi^\prime, \tilde{m}, \chi, q^{\prime\prime}) \in \tau^\prime$. Let $\bar{\phi} = \psi \wedge (\phi - \psi)$, where $\phi - \psi$ contains those literals $L$ from $\phi$ for which neither $L$ nor its negation appears in $\psi$. If $\bar{\phi}$ and $\phi^\prime$ are compatible, we may combine $t$ and $t^\prime$ to the new transition $t^\prime \circ t = (p, \bar{\phi} \wedge \phi^\prime, \tilde{m}, \chi, q^{\prime\prime})$. Then we can replace $\tau$ by $\tau^\prime \circ \tau$ containing all those composed transition $t^\prime \circ t$ with $t \in \tau$ and $t^\prime \in \tau^\prime$. This replacement can be iterated until there are no more transition sets with $\epsilon$-transitions, which defines $\delta_p$.

Then take transitions $t = (q,\epsilon,q^\prime) \in \tau$ and $t^\prime = (q^\prime, \tilde{m}, q^{\prime\prime}) \in \tau^\prime$. We may combine $t$ and $t^\prime$ to the new transition $t^\prime \circ t = (p, \tilde{m}, q^{\prime\prime})$.

Then we can replace $\tau$ by $\tau^\prime \circ \tau$ containing all those composed transition $t^\prime \circ t$ with $t \in \tau$ and $t^\prime \in \tau^\prime$. This replacement can be iterated until there are no more transition sets with $\epsilon$-transitions, which defines $\delta_p$.

\end{description}

\begin{example}\label{bsp-projection}

\ Consider the choreography from Example \ref{bsp-choreo}. For peer 1 the first step above gives rise to the $\epsilon$-extended control-state machine in Figure \ref{fig-projection}.

\begin{figure}[htb]
\begin{center}
\includegraphics[scale=0.6]{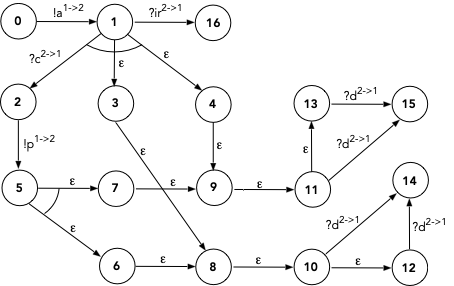}
\caption{Projection of a Choreography\label{fig-projection}}
\end{center}
\end{figure}

In step 2 all $\epsilon$-transitions are eliminated. The resulting control-state machine is equivalent to the one in Figure \ref{fig-p2p_peer1}, if we rename the control-states: 5 becomes 15, 3 becomes 5, and 4 becomes 14. 

For peer 2 the first step above gives already rise to the control-state machine in Figure \ref{fig-p2p_peer2} without $\epsilon$-transitions, so step 2 is obsolete.

\end{example}

This construction suggests the following terminology of a choreography-defined P2P system.

\begin{definition}\label{def-choreo-def}

A P2P system $\{ \mathcal{P}_p \}_{p \in P}$ is {\em choreography-defined} iff there exists a CSC $\mathcal{C}$ such that the projected P2P system $\{ \mathcal{C}_p \}_{p \in P}$ is equivalent to $\{ \mathcal{P}_p \}_{p \in P}$.

\end{definition}

\begin{proposition}\label{prop-csc-p2p}

The trace language $\mathcal{L}(\mathcal{C})$ of the choreography $\mathcal{C}$ is equal to the trace language $\mathcal{L}_0$ defined by the rendez-vous semantics of the P2P system $\{ \mathcal{C}_p \}_{p \in P}$ and the P2P system $\{ \mathcal{C}_p^\prime \}_{p \in P}$ with $\epsilon$-transitions.

\end{proposition}

\begin{proof}

We first show $\mathcal{L}(\mathcal{C}) \subseteq \mathcal{L}_0(\{ \mathcal{C}_p \}_{p \in P})$ and $\mathcal{L}(\mathcal{C}) \subseteq \mathcal{L}_0(\{ \mathcal{C}_p^\prime \}_{p \in P})$. For $m_0 \dots m_\ell \in \mathcal{L}(\mathcal{C})$ we can write again $m_0 \dots m_\ell = \hat{m}_0 \dots \hat{m}_k$ as in the proof of Proposition \ref{prop-p2p-csc}, corresponding to a run $S_0, \dots, S_{k+1}$ with $S_{i+1} = S_i + \Delta_i$ for $i = 0, \dots, k$.

Here we have $\Delta_i = \{ (ct,q_{i+1}) \} = \bigcup_{t \in \tau_i, t \in E_{\mathcal{C}}(S_i)} \Delta_t(S_i)$. These transitions $t$ have the form $(q_i, p \stackrel{m}{\rightarrow} p^\prime, q_{i+1})$, so they define projections $(q_i, !m^{p \rightarrow p^\prime}, q_{i+1}) \in \tau_p$ and $(q_i, ?m^{p \rightarrow p^\prime}, q_{i+1}) \in \tau_{p^\prime}$ plus $\epsilon$-transitions in $\tau_{p^0}$ for $p^0 \neq p$, $p^0 \neq p^\prime$.

In $\{ \mathcal{C}_p^\prime \}_{p \in P}$ with rendez-vous semantics these projected transitions in $\tau_p$ with $p \in P$ ignoring $\epsilon$-transitions are simultaneously enabled, so we get a run $\bar{S}_0, \bar{S}_1, \dots, \bar{S}_{k+1}$ of $\{ \mathcal{C}_p^\prime \}_{p \in P}$ in rendez-vous semantics with the same sets of messages $M_i$, hence $m_0 \dots m_\ell \in \mathcal{L}_0(\{ \mathcal{C}_p^\prime \}_{p \in P})$. As we can ignore $\epsilon$-transitions, the same applies for $\{ \mathcal{C}_p \}_{p \in P}$, which shows $m_0 \dots m_\ell \in \mathcal{L}_0(\{ \mathcal{C}_p \}_{p \in P})$.

To show the opposite inclusions we start from $m_0 \dots m_\ell \in \mathcal{L}_0(\{ \mathcal{C}_p \}_{p \in P})$ (or $m_0 \dots m_\ell \in \mathcal{L}_0(\{ \mathcal{C}_p^\prime \}_{p \in P})$) and consider a run $\bar{S}_0, \bar{S}_1, \dots, \bar{S}_{k+1}$ of $\{ \mathcal{C}_p^\prime \}_{p \in P}$ (or of $\{ \mathcal{C}_p^\prime \}_{p \in P}$) in rendez-vous semantics with sets of messages $M_i$. As in the proof of Proposition \ref{prop-p2p-csc} we can assume without loss of generality to always have $a(i,p) = i$ (see Definition \ref{def-run}). The corresponding update sets are defined by transition sets $\tau_p$ with $p \in P$ containing simultaneously transitions of the form $(q_i, !m^{p \rightarrow p^\prime}, q_{i+1}) \in \tau_p$ and $(q_i, ?m^{p \rightarrow p^\prime}, q_{i+1}) \in \tau_{p^\prime}$. That is, the enabled transitons are projections of transitions $(q_i, p \stackrel{m}{\rightarrow} p^\prime, q_{i+1})$ of $\mathcal{C}$. This defines a run $S_0, \dots, S_{k+1}$ of $\mathcal{C}$ with $S_{i+1} = S_i + \Delta_i$ for $i = 0, \dots, k$, in which we have $\Delta_i = \{ (ct,q_{i+1}) \} = \bigcup_{t \in \tau_i, t \in E_{\mathcal{C}}(S_i)} \Delta_t(S_i)$. This run produces the same message sets $M_i$ as $\bar{S}_0, \bar{S}_1, \dots, \bar{S}_{k+1}$, hence $m_0 \dots m_\ell \in \mathcal{L}(\mathcal{C})$.\qed

\end{proof}

Proposition \ref{prop-csc-p2p} shows that a choreography-defined P2P systems and choreographies mutually define each other, and the language defined by a choreography is the rendez-vous language of the corresponding P2P system. This justifies the following notion of realisability.

\begin{definition}\label{def-realisable}

A CSC $\mathcal{C}$ is called {\em realisable} if and only if its projected P2P system $\{ \mathcal{C}_p \}_{p \in P}$ is synchronisable.

\end{definition}

Clearly, the notion of realisability depends on p2p, queue or mailbox semantics.

\subsection{Synchronisability for Choreography-Defined P2P Systems}

Definition \ref{def-sync} also defines the weaker notion of language synchronisability, so we could analogously define ``language-realisability'' for CSCs. However, for the restricted case of choreographies defined by FSMs we showed in \cite[Prop.5]{schewe:foiks2020} that language synchronisability and synchronisability coincide for choreography-defined P2P systems. This result also holds for our generalised CSCs.

\begin{proposition}\label{prop-lang-realisable}

A choreography-defined P2P system $\mathcal{P} = \{ \mathcal{P}_p \}_{p \in P}$ is synchronisable with respect to p2p, queue or mailbox semantics, respectively, if and only if it is language synchronisable with respect to the same semantics.

\end{proposition}

\begin{proof}

According to Definition \ref{def-sync} and the remarks following it we have to show $\mathcal{L}_\omega(\mathcal{P}) \subseteq \hat{\mathcal{L}}_\omega(\mathcal{P})$ for $\omega \in \{ p2p, q, m \}$. So let $w \in \mathcal{L}_\omega(\mathcal{P})$, say $w = i_1 \stackrel{m_1}{\rightarrow} j_1 \dots i_\ell \stackrel{m_\ell}{\rightarrow} j_\ell$. Let $w^! = !m_1^{i_1 \rightarrow j_1} \dots !m_\ell^{i_\ell \rightarrow j_\ell}$ and $w^? = ?m_1^{i_1 \rightarrow j_1} \dots ?m_\ell^{i_\ell \rightarrow j_\ell}$ be the corresponding sequences of sending and receiving events, respectively.

Then there exists an interleaving $w^\prime$ of $w^!$ and a prefix of $w^?$, say $w^\prime = w_0 \dots w_k$, where each $w_i$ is a sequence of sending and receiving events and $!m_x^{i_x \rightarrow j_x}$ precedes $?m_x^{i_x \rightarrow j_x}$ in $w^\prime$---if the latter one appears in $w^\prime$ at all---and there exists a run $S_0, \dots, S_{k+1}$ of $\mathcal{P}$ in p2p, queue or mailbox semantics with $S_{i+1} = S_i + \Delta_i$ for $i = 0, \dots, k$, in which the update sets satisfy $\Delta_i = \bigcup_{p \in P_i} \Delta_{a(i,p)}$ with
\[ \Delta_{a(i,p)} = \Delta_{\tau_{i,p}}(S_{a(i,p)}) = \bigcup_{ t \in \tau_{i,p}, t \in E_p(S_{a(i,p)})} \Delta_t(S_{a(i,p)}) \in \boldsymbol{\Delta}_p (S_{a(i,p)}) \; . \]

If $!m_x^{i_x \rightarrow j_x}$ appears in $w_i$, then with $p = i_x$ there exists an enabled transition $t \in \tau_{i,p}$ of the form $t = (q_t, !m_x^{i_x \rightarrow j_x}, q_t^\prime)$. In this case we have $((ctl, p), q_t^\prime) \in \Delta_t(S_{a(i,p)})$ and according to Definition \ref{def-transition-update} depending on the semantics one of the following holds:
\begin{gather}
( (channel, (i_x,j_x)), val_S((channel,(p,j_x))) \oplus [ i_x \stackrel{m_x}{\rightarrow} j_x ] ) \in \Delta_t(S_{a(i,p)}) \tag{p2p}\\
( (queue, j_x), val_S((queue,j_x)) \oplus [ i_x \stackrel{m_x}{\rightarrow} j_x ] ) \in \Delta_t(S_{a(i,p)}) \tag{queue} \\
( (mailbox, j_x), val_S((mailbox,j_x)) \uplus [ i_x \stackrel{m_x}{\rightarrow} j_x ] ) \in \Delta_t(S_{a(i,p)}) \tag{mailbox}
\end{gather}

Analogously, if $?m_x^{i_x \rightarrow j_x}$ appears in $w_i$, then with $p = j_x$ there exists an enabled transition $t \in \tau_{i,p}$ of the form $t = (q_t, ?m_x^{i_x \rightarrow j_x}, q_t^\prime)$. In this case we have $((ctl, p), q_t^\prime) \in \Delta_t(S_{a(i,p)})$ and according to Definition \ref{def-transition-update} depending on the semantics one of the following holds:
\begin{gather}
( (channel, (i_x,j_x)), val_S((channel,(i_x,p))) \ominus [ i_x \stackrel{m_x}{\rightarrow} j_x ] ) \in \Delta_t(S_{a(i,p)}) \tag{p2p} \\
( (queue, j_x), val_S((queue,p)) \ominus [ i_x \stackrel{m_x}{\rightarrow} j_x ] ) \in \Delta_t(S_{a(i,p)}) \tag{queue} \\
( (mailbox, j_x), val_S((mailbox,p)) - [ i_x \stackrel{m_x}{\rightarrow} j_x ] ) \in \Delta_t(S_{a(i,p)}) \tag{mailbox}
\end{gather}
In the first two cases $i_x \stackrel{m_x}{\rightarrow} j_x$ must have been the front element in the FIFO queue. Then in $S_{k+1}$ the $channel$, $queue$ or $mailbox$ locations, respectively, will still contain the messages $i_x \stackrel{m_x}{\rightarrow} j_x$, for which the reciving event $?m_x^{i_x \rightarrow j_x}$ does not occur in $w^\prime$.

As $\mathcal{P}$ is choreography-defined, all transitions are projections from a CSC $\mathcal{C} = (Q, q_0, M, \delta)$. Hence in both cases above---$t = (q_t, !m_x^{i_x \rightarrow j_x}, q_t^\prime) \in \tau_p \in \delta_p$ with $p = i_x$ or $t = (q_t, ?m_x^{i_x \rightarrow j_x}, q_t^\prime) \in \tau_p \in \delta_p$ with $p = j_x$---we obtain $(q_t, i_x \stackrel{m_x}{\rightarrow} j_x, q_t^\prime) \in \tau \in \delta$.

As we have $\mathcal{L}_0(\mathcal{P}) = \mathcal{L}_\omega(\mathcal{P})$ for $\omega \in \{ p2p, q, m \}$, this defines a run of $\mathcal{C}$ and due to Proposition \ref{prop-csc-p2p} we obtain a run $\bar{S}_0, \dots, \bar{S}_{k+1}$ of $\mathcal{P}$ in rendez-vous semantics with the same message sets $M_i$ ($i=0,\dots,k$) as the run $S_0, \dots, S_{k+1}$. Let $\bar{S}_{i+1} = \bar{S}_i + \bar{\Delta}_i$ in this run, i.e. whenever a sending transition is used in $\Delta_i$, the corresponding receive transition is used in $\bar{\Delta}_i$ at the same time.

The run $\bar{S}_0, \dots, \bar{S}_{k+1}$ also defines a run of $\mathcal{P}$ in p2p, queue or mailbox semantics, respectively, if the updates to $channel$, $queue$ and $mailbox$ locations are added. In this run we have that one of $val_{\bar{S}_{k+1}}(channel,(i,j)) = []$, $val_{\bar{S}_{k+1}}(queue,j) = []$ or $val_{\bar{S}_{k+1}}(mailbox,j) = \langle\rangle$ holds for all peers $i,j$. This shows $w \in \hat{\mathcal{L}}_\omega(\mathcal{P})$ for the corresponding $\omega \in \{ p2p, q, m \}$.\qed

\end{proof}

\section{Characterisation of Realisability}\label{sec:realisability}

We now investigate realisability of CSCs. From Proposition \ref{prop-csc-p2p} we know that CSCs are equivalent to their projected P2P systems concerning rendez-vous semantics in the sense of defining the same trace language, so a CSC prescribes the behaviour of a P2P system capturing communication in a synchronous way, i.e. sending and receiving of messages are always synchronised. We have to show the synchronisability of these P2P systems, which means to show that the semantics expressed by the trace languages is preserved, when a more realistic asynchronous behaviour using channels, queues or mailboxes is taken into account. As a consequence of Proposition \ref{prop-lang-realisable} it suffces to concentrate on language synchronisability.

Not all P2P systems are choreography-defined, and not all choreographies will be realisable. We will now derive necessary and sufficient conditions for CSCs that guarantee realisability. For this we will generalise the sequence and choice conditions from \cite{schewe:foiks2020}.

\subsection{The Sequence Condition}

In \cite[Prop.6]{schewe:foiks2020} we proved that if a choreography defined by an FSM is realisable, then any two messages that follow each other in a sequence must either be independent, i.e. their order can be swapped, or the sender of the second message must be the same as the sender or receiver of the first message. The first alternative condition was a workaround for parallelism, and the second condition expressed that messages in a sequence must either have the same sender or the receiver of a message must be the sender of the next one. We will now show that in a slightly generalised way this condition must also be satisfied for CSCs in order to obtain realisability. 

\begin{definition}\label{def-sequence}

A CSC $\mathcal{C} = (Q, q_0, M, \delta)$ is said to satisfy the {\em sequence condition} if and only if for any two transitions $(q_1, i \stackrel{m_1}{\rightarrow} j, q_2) \in \tau \in \delta$ and $(q_2, k \stackrel{m_2}{\rightarrow} \ell, q_3) \in \tau^\prime \in \delta$ we have that $(q_1, k \stackrel{m_2}{\rightarrow} \ell, q_2) \in \tau \in \delta$ and $(q_2, i \stackrel{m_1}{\rightarrow} j, q_3) \in \tau^\prime \in \delta$ hold or $k \in \{ i,j \}$ holds.

\end{definition}

\begin{proposition}\label{prop-sequence}

If $\mathcal{C} = (Q, q_0, M, \delta)$ is a realisable CSC with respect to p2p, queue or mailbox semantics, then $\mathcal{C}$ satisfies the sequence condition.

\end{proposition}

\begin{proof}

Without loss of generality we can assume that $\mathcal{C}$ only contains reachable states. Now assume that the sequence condition is violated. Then there exist transitions $(q_1, i \stackrel{m_1}{\rightarrow} j, q_2) \in \tau \in \delta$ and $(q_2, k \stackrel{m_2}{\rightarrow} \ell, q_3) \in \tau^\prime \in \delta$, but not $(q_1, k \stackrel{m_2}{\rightarrow} \ell, q_2) \in \tau$, such that $i,j,k \in P$ are pairwise different. Consider the projections in the P2P system $\mathcal{P} = \{ \mathcal{C}_p \mid p \in P \}$:
\begin{gather*}
(q_1, !m_1^{i \rightarrow j}, q_2) \in \tau_1 \in \delta_i \qquad\qquad
(q_1, ?m_1^{i \rightarrow j}, q_2) \in \tau_2 \in \delta_j \\
(q_2, !m_2^{k \rightarrow \ell}, q_3) \in \tau_3 \in \delta_k \qquad\qquad
(q_2, ?m_2^{k \rightarrow \ell}, q_3) \in \tau_4 \in \delta_\ell
\end{gather*}

Then $\mathcal{L}_0(\mathcal{P})$ contains sequences $w w_1 w_2 w^\prime$, where 
$i \stackrel{m_1}{\rightarrow} j$ appears in $w_1$, $k \stackrel{m_2}{\rightarrow} \ell$ appears in $w_2$, and all $w_1^\prime$, $w_2^\prime$ resulting from $w_1$, $w_2$ by permutation of the messages give rise to sequences $w w_1^\prime w_2^\prime w^\prime \in \mathcal{L}_0()$, but $k \stackrel{m_2}{\rightarrow} \ell$ does not appear in $w_1$.

The sequence $w w_1 w_2 w^\prime$ corresponds to a run $S_0, \dots, S_{i-1}, S_i, S_{i+1}, \dots, S_k$ of $\mathcal{P}$ in p2p, queue or mailbox semantics, respectively, with $w_1 \in \hat{M}_{i-1}$ and $w_2 \in \hat{M}_i$. In particular, we have $val_{S_{i-1}}(ctl,i) = q_1$, $val_{S_{i-1}}(ctl,j) = q_1$ and $val_{S_i}(ctl,i) = val_{S_i}(ctl,j) = val_{S_i}(ctl,k) = q_2$. As $\mathcal{P}$ is choreography-defined, we can assume without loss of grenerality that $k \notin P_{i-1}$ holds, which implies $val_{S_{i-1}}(ctl,k) = q_2$, i.e. there is no transition with peer $k$ in the $i$'th step of the run.

We may therefore add $\tau_3$ to define the update set in the $i$'th step, which creates another run, in which $M_{i-1}$ is extended, in particular $k \stackrel{m_2}{\rightarrow} \ell \in M_{i-1}$. With this run we obtain $w \bar{w}_1 \bar{w}_2 w^\prime \in \mathcal{L}_\omega(\mathcal{P})$ for $\omega \in \{ p2p, q, m \}$, and $k \stackrel{m_2}{\rightarrow} \ell$ appears in $\bar{w}_1$. This contradicts our assumptions on $w w_1 w_2 w^\prime$.\qed

\end{proof}

\subsection{The Choice Condition}

In \cite[Prop.7]{schewe:foiks2020} we further proved that if a choreography defined by an FSM is realisable, then any two messages subject to a choice must be independent or have the same sender. The second condition expresses that messages in a choice can only refer to the same sender. We will now show that in a slightly generalised way this condition must also be satisfied for CSCs in order to obtain realisability. 

\begin{definition}\label{def-choice}

A CSC $\mathcal{C} = (Q, q_0, M, \delta)$ is said to satisfy the {\em choice condition} if and only if for any two transitions $(q_1, i \stackrel{m_1}{\rightarrow} j, q_2) \in \tau \in \delta$ and $(q_1, k \stackrel{m_2}{\rightarrow} \ell, q_3) \in \tau^\prime \in \delta$ we have that $(q_1, k \stackrel{m_2}{\rightarrow} \ell, q_2) \in \tau \in \delta$ and $(q_1, i \stackrel{m_1}{\rightarrow} j, q_3) \in \tau^\prime \in \delta$ hold or $k = i$ holds.

\end{definition}

\begin{proposition}\label{prop-choice}

If $\mathcal{C} = (Q, q_0, M, \delta)$ is a realisable CSC with respect to p2p, queue or mailbox semantics, then $\mathcal{C}$ satisfies the choice condition.

\end{proposition}

\begin{proof}

As in the proof of Proposition \ref{prop-sequence} we can assume without loss of generality that $\mathcal{C}$ only contains reachable states. Now assume that the choice condition is violated. Then there exist transitions $(q_1, i \stackrel{m_1}{\rightarrow} j, q_2) \in \tau \in \delta$ and $(q_1, k \stackrel{m_2}{\rightarrow} \ell, q_3) \in \tau^\prime \in \delta$, but $(q_1, k \stackrel{m_2}{\rightarrow} \ell, q_2) \notin \tau$ and $(q_1, i \stackrel{m_1}{\rightarrow} j, q_3) \notin \tau^\prime$, such that $i \neq k$ holds. Consider the projections in the P2P system $\mathcal{P} = \{ \mathcal{C}_p \mid p \in P \}$:
\begin{gather*}
(q_1, !m_1^{i \rightarrow j}, q_2) \in \tau_1 \in \delta_i \qquad\qquad
(q_1, ?m_1^{i \rightarrow j}, q_2) \in \tau_2 \in \delta_j \\
(q_1, !m_2^{k \rightarrow \ell}, q_3) \in \tau_3 \in \delta_k \qquad\qquad
(q_1, ?m_2^{k \rightarrow \ell}, q_3) \in \tau_4 \in \delta_\ell
\end{gather*}

Then $\mathcal{L}_0(\mathcal{P})$ contains sequences $w w_1 w^\prime$ and $w w_2 w^{\prime\prime}$, where $i \stackrel{m_1}{\rightarrow} j$ appears in $w_1$, $k \stackrel{m_2}{\rightarrow} \ell$ appears in $w_2$, and all $w_1^\prime$, $w_2^\prime$ resulting from $w_1$, $w_2$ by permutation of the messages give rise to sequences $w w_1^\prime w^\prime, w w_2^\prime w^{\prime\prime} \in \mathcal{L}_0()$, but $k \stackrel{m_2}{\rightarrow} \ell$ does not appear in $w_1$.

These sequences correspond to runs $S_0, \dots, S_i, S_{i+1}, \dots, S_{k_1}$ and $S_0, \dots, S_i, S_{i+1}^\prime$, $\dots, S_{k_2}^\prime$ of $\mathcal{P}$ in p2p, queue or mailbox semantics, respectively, with $w_1 \in \hat{M}_{i,1}$ and $w_2 \in \hat{M}_{i,2}$ using the message sets in the step $(i+1)$ in the two runs. In particular, we have $val_{S_i}(ctl,i) = val_{S_i}(ctl,k) = q_1$, $val_{S_{i+1}}(ctl,i) = q_2$ and $val_{S_{i+1}}(ctl,k) = q_3$. 

As $\mathcal{P}$ is choreography-defined and $k \neq i$ holds, we can add $\tau_3$ to define the update set in the step $(i+1)$ of the first of these two runs. This creates another run, in which $M_i$ is extended, in particular $k \stackrel{m_2}{\rightarrow} \ell \in M_i$. With this run we obtain $w \bar{w}_1 w^\prime \in \mathcal{L}_\omega(\mathcal{P})$ for $\omega \in \{ p2p, q, m \}$, where $\bar{w}_1$ extends $w_1$ and $k \stackrel{m_2}{\rightarrow} \ell$ appears in $\bar{w}_1$. This contradicts our assumptions on $w w_1 w^\prime$.\qed

\end{proof}

\subsection{Sufficient Conditions for Realisability}

In \cite[Thm.1]{schewe:foiks2020} we also proved that the two necessary conditions for realisability, the sequence and the choice conditions, together are also sufficient. We will now generalise this result for CSCs, which gives the main result of this article.

\begin{theorem}\label{thm-realisability}

A CSC $\mathcal{C} = (Q, q_0, M, \delta)$ is realisable with respect to p2p, queue or mailbox semantics, if and only if it satisfies the sequence and the choice conditions.

\end{theorem}

\begin{proof}

The necessity of the sequence and choice conditions was already proven in Propositions \ref{prop-sequence} and \ref{prop-choice}. For the sufficience let $\mathcal{P} = \{ \mathcal{C}_p \}_{p \in P}$ be the P2P system defined by the projection of the CSC $\mathcal{C}$. 

It suffices to show $\mathcal{L}_\omega(\mathcal{P}) \subseteq \mathcal{L}_0(\mathcal{P})$, where $\omega \in \{ p2p, q, m \}$. As trace languages are prefix-closed, we proceed by induction on the length of sequences of messages. The induction base for the empty sequence $\epsilon$ is trivial.

Now assume $w = w^\prime \bar{w} \in \mathcal{L}_\omega(\mathcal{P})$. Then $w$ is produced by a run $S_0, \dots, S_{k+1}$ of $\mathcal{P}$ in p2p, queue or mailbox semantics, respectively, with message sets $M_i$ ($i=0,\dots,k$) such that $\bar{w} \in \bar{M}_k$ and $w^\prime \in \mathcal{L}_\omega(\mathcal{P})$. By induction we have $w^\prime \in \mathcal{L}_0(\mathcal{P})$, so there is also a run $S_0^\prime, \dots, S_{k^\prime+1}^\prime$ of $\mathcal{P}$ in rendez-vous semantics corresponding to $w^\prime$ with associated message sets $M_i^\prime$ ($i=0,\dots,k^\prime$).

\paragraph{\textbf{Case 1.}}

Assume $w^\prime = \epsilon$. If we have $S_1 = S_0 + \Delta_0$, then $\Delta_0 = \bigcup_{p \in P_0} \Delta_{\tau_p}(S_0)$ holds with $\tau_p \in \delta_p$, and $\Delta_{\tau_p}(S_0) = \bigcup_{t \in \tau_p, t \in E_p(S_0)} \Delta_t(S_0)$.

The enabled transitions in $\tau_i$ must be sending transitions of the form $(q_{i,0}, !m^{i \rightarrow j} q_{i,1})$; there cannot be any receiving transitions, because in p2p, queue or mailbox semantics messages first must be sent, before they can be received. As $\mathcal{P}$ is choreography-defined, each transition set $\tau_i$ is a projection of some $\tau^{(i)} \in \delta$, so we have transitions $(q_{i,0}, i \stackrel{m}{\rightarrow} j, q_{i,1}) \in \tau^{(i)}$ and $q_{i,0} = q_0$ for all $i \in P$.

Due to the choice condition either all these transitions occur in the same $\tau \in \delta$ or their messages all have the same sender $k \in P$. 

In the former case, whenever $(q_{i,0}, !m^{i \rightarrow j}, q_{i,1})$ is enabled in rendez-vous semantics, then also $(q_{i,0}, ?m^{i \rightarrow j}, q_{i,1})$ is enabled, which implies $\bar{w} \in \mathcal{L}_0(\mathcal{P})$. In the latter case, whenever $(q_{i,0}, !m^{i \rightarrow j}, q_{i,1})$ is enabled in rendez-vous semantics, then $i = k$ and $j \neq k$ and $(q_{j,0}, ?m^{i \rightarrow j}, q_{j,1})$ is also enabled in rendez-vous semantics, which implies again $\bar{w} \in \mathcal{L}_0(\mathcal{P})$.

\paragraph{\textbf{Case 2.}}

Now assume $w^\prime \neq \epsilon$. As the run $S_0, \dots, S_k$ in p2p, queue or mailbox semantics corresponds to $w^\prime \in \mathcal{L}_\omega(\mathcal{P})$, let $S_0, \dots, S_{k-1}$ correspond to $w^{\prime\prime}$, i.e. $w^\prime = w^{\prime\prime} \ddot{w}$. We distinguish two subcases.

\paragraph{\textbf{Case 2a.}}

Assume $w^{\prime\prime} \bar{w} \in \mathcal{L}_\omega(\mathcal{P})$. Then by induction $w^{\prime\prime} \bar{w} \in \mathcal{L}_0(\mathcal{P})$.

As also $w^\prime = w^{\prime\prime} \ddot{w} \in \mathcal{L}_0(\mathcal{P})$ holds, the run $S_0^\prime, \dots, S_{k^\prime}^\prime$ of $\mathcal{P}$ in rendez-vous semantics corresponding to $w^{\prime\prime}$ has two possible continuations $S_0^\prime, \dots, S_{k^\prime+1}^\prime$ and $S_0^\prime, \dots, S_{k^\prime}^{\prime\prime}$ corresponding to $w^\prime$ and $w^{\prime\prime} \bar{w}$, respectively. Let the update sets in the last state transitions in these runs be $\Delta_{k^\prime}^\prime$ and $\Delta_{k^\prime}^{\prime\prime}$ , respectively, i.e. $S_{k^\prime+1}^\prime = S_{k^\prime}^\prime + \Delta_{k^\prime}^\prime$ and $S_{k^\prime+1}^{\prime\prime} = S_{k^\prime}^\prime + \Delta_{k^\prime}^{\prime\prime}$.

We have $\Delta_{k^\prime}^\prime = \bigcup_{p \in P_{k^\prime}} \Delta_{a(k^\prime,p)}$ with update sets \[ \Delta_{a(k^\prime,p)} = \Delta_{\tau_p}(S_{a(k^\prime,p)}^\prime) = \bigcup_{{t \in \tau_p} \atop {t \in E_p(S_{a(k^\prime,p)}^\prime)}} \Delta_t(S_{a(k^\prime,p)}^\prime) \in \boldsymbol{\Delta}_p(S_{a(k^\prime,p)}^\prime) \; , \]
and analogously, $\Delta_{k^\prime}^{\prime\prime} = \bigcup_{p \in P_{k^\prime}^\prime} \Delta_{a^\prime(k^\prime,p)}$ with update sets \[ \Delta_{a^\prime(k^\prime,p)} = \Delta_{\tau_p^\prime}(S_{a^\prime(k^\prime,p)}^\prime) = \bigcup_{{t \in \tau_p^\prime} \atop {t \in E_p(S_{a^\prime(k^\prime,p)}^\prime)}} \Delta_t(S_{a^\prime(k^\prime,p)}^\prime) \in \boldsymbol{\Delta}_p(S_{a^\prime(k^\prime,p)}^\prime) \; . \]

An enabled transition $t$ in $\tau_p$ or $\tau_p^\prime$ has either the form $(q_p, !m_1^{p \rightarrow p^\prime}, q_p^\prime)$ or $(q_p, ?m_2^{p^\prime \rightarrow p}, q_p^{\prime\prime})$. As $\mathcal{P}$ is choreography-defined, we also have $(q_p, p \stackrel{m_1}{\rightarrow} p^\prime, q_p^\prime) \in \tau_1 \in \delta$ and $(q_p, p^\prime \stackrel{m_2}{\rightarrow} p, q_p^{\prime\prime}) \in \tau_2 \in \delta$. Furthermore, if a transition with sending event $!m_1^{p \rightarrow p^\prime}$ appears in $\tau_p$ (or in $\tau_p^\prime$, respectively), then also the corresponding receiving event $?m_1^{p \rightarrow p^\prime}$ appears in $\tau_{p^\prime}$ (or in $\tau_{p^\prime}^\prime$, respectively).

The messages $i \stackrel{m_1}{\rightarrow} j$ in $\ddot{w}$ correspond to pairs of transitions $(q_i, !m_1^{i \rightarrow j}, q_i^\prime) \in \tau_i$ and $(q_j, ?m_1^{i \rightarrow j}, q_j^\prime) \in \tau_j$ and thus to $(q_i, i \stackrel{m_1}{\rightarrow} j, q_i^\prime) \in \tau \in \delta$ with $q_i = q_j$ and $q_i^\prime = q_j^\prime$. Likewise, the messages $k \stackrel{m_2}{\rightarrow} \ell$ in $\bar{w}$ correspond to pairs of transitions $(q_k, !m_2^{k \rightarrow \ell}, q_k^{\prime\prime}) \in \tau_k^\prime$ and $(q_\ell, ?m_2^{k \rightarrow \ell}, q_\ell^{\prime\prime}) \in \tau_\ell^\prime$ and thus to $(q_k, k \stackrel{m_2}{\rightarrow} \ell, q_k^{\prime\prime}) \in \tau^\prime \in \delta$ with $q_k = q_\ell$ and $q_k^{\prime\prime} = q_\ell^{\prime\prime}$.

If all control states $q_i$ appearing in $\tau_i$ and all $q_k$ appearing in $\tau_k^\prime$ are pairwise different, we can build a new update set
\[ \bar{\Delta}_{k^\prime} = \bigcup_{p \in P_{k^\prime}} \Delta_{a(k^\prime,p)} \cup \bigcup_{p \in P_{k^\prime}^\prime} \Delta_{a^\prime(k^\prime,p)} \]
in $S_{k^\prime}^\prime$, which gives rise to a run $S_0^\prime, \dots, S_{k^\prime}^\prime, \bar{S}_{k^\prime + 1}$ in rendez-vous semantics with $\bar{S}_{k^\prime + 1} = S_{k^\prime}^\prime + \bar{\Delta}_{k^\prime}$. The message set $M_{k^\prime}$ in the last step defines $\ddot{w} \bar{w} \in \bar{M}_{k^\prime}$, hence $w^\prime \bar{w} = w^{\prime\prime} \ddot{w} \bar{w} \in \mathcal{L}_0(\mathcal{P})$.

However, if the sets of control states $q_i$ and $q_k$, respectively, overlap, the control-state choreography contains choices $(q, i \stackrel{m_1}{\rightarrow} j, q^\prime) \in \tau \in \delta$ and $(q, k \stackrel{m_2}{\rightarrow} \ell, q^{\prime\prime}) \in \tau^\prime \in \delta$. As $\mathcal{C}$ satisfies the choice condition, we either have $(q, k \stackrel{m_2}{\rightarrow} \ell, q^\prime) \in \tau$ and $(q, i \stackrel{m_1}{\rightarrow} j, q^{\prime\prime}) \in \tau^\prime$ or $k = i$.

For the former cases we can take a subset $P_{k^\prime}^{\prime\prime} = \{ k \in P_{k^\prime}^\prime \mid (q_k, !m_2^{k \rightarrow \ell}, q_k^\prime) \in \tau_k \}$, which allows us to build the update set
\[ \bar{\Delta}_{k^\prime} = \bigcup_{p \in P_{k^\prime}} \Delta_{a(k^\prime,p)} \cup \bigcup_{p \in P_{k^\prime}^{\prime\prime}} \Delta_{a^\prime(k^\prime,p)} \]
in $S_{k^\prime}^\prime$. We obtain a run $S_0^\prime, \dots, S_{k^\prime}^\prime, \bar{S}_{k^\prime + 1}$ of $\mathcal{P}$ in rendez-vous semantics with $\bar{S}_{k^\prime + 1} = S_{k^\prime}^\prime + \bar{\Delta}_{k^\prime}$ and $w^\prime \bar{w}^\prime \in \mathcal{L}_0(\mathcal{P})$, where $\bar{w}^\prime$ is a subsequence of $\bar{w}$. 

For the case $k = i$ we have $(q_i, !m_1^{i \rightarrow j}, q_i^\prime) \in \tau_i$ and $(\bar{q}_i, !m_2^{i \rightarrow \ell}, q_i^{\prime\prime}) \in \tau_i^\prime$. As the sequence $i \stackrel{m_1}{\rightarrow} j \; i \stackrel{m_2}{\rightarrow} \ell$ appears in $\bar{w}$ in some order, there must exist a sequence of control states $q_{i,1} , \dots, q_{i,m}$ with $q_{i,1} = q_i^\prime$, $q_{i,m} = \bar{q}_i$ such that $(q_{i,x}, ?m_x^{j \rightarrow i}, q_{i,x+1}) \in \tau^{(x)} \in \delta_i$ holds for $i = 1, \dots, m-1$. That is, there is a sequence of enabled transitions between $(q_i, !m_1^{i \rightarrow j}, q_i^\prime)$ and $(\bar{q}_i, !m_2^{i \rightarrow \ell}, q_i^{\prime\prime})$. Due to the requirement on $w^\prime \bar{w} \in \mathcal{L}_\omega(\mathcal{P})$ only receiving events can appear in this sequence, and due to the sequence condition the sender must always be $j$. However, the corresponding sending events must then also appear in enabled transitions following $(q_i, !m_1^{i \rightarrow j}, q_i^\prime)$, which contradicts that $i \stackrel{m_1}{\rightarrow} j \; i \stackrel{m_2}{\rightarrow} \ell$ appears in $\bar{w}$, unless we have $m = 1$, i.e. $q_i^\prime = \bar{q}_i$.

Then $P_{k^\prime}^\prime - P_{k^\prime}^{\prime\prime}$ defines an update set $\bar{\Delta}_{k^\prime + 1} = \bigcup_{p \in P_{k^\prime}^\prime - P_{k^\prime}^{\prime\prime}} \Delta_{a^\prime(k^\prime,p)}$ in $\bar{S}_{k^\prime + 1}$ and an extended run $S_0^\prime, \dots, S_{k^\prime}^\prime, \bar{S}_{k^\prime + 1}, \bar{S}_{k^\prime+2}$ in rendez-vous semantics with $\bar{S}_{k^\prime+2} = \bar{S}_{k^\prime + 1} + \bar{\Delta}_{k^\prime + 1}$, andv $w^\prime \bar{w}^\prime \bar{w}^{\prime\prime} = w^\prime \bar{w} \in \mathcal{L}_0(\mathcal{P})$, as $\bar{w}^{\prime\prime}$ contains those messages from $\bar{w}$ not in $\bar{w}^\prime$, i.e. $\bar{w}^\prime \bar{w}^{\prime\prime} = \bar{w}$.

\paragraph{\textbf{Case 2b.}}

Now assume that $w^{\prime\prime} \bar{w} \notin \mathcal{L}_\omega(\mathcal{P})$ holds. For $w^\prime = w^{\prime\prime} \ddot{w} \in \mathcal{L}_0(\mathcal{P})$ we obtain a run $S_0^\prime, \dots, S_{k^\prime}^\prime$ of $\mathcal{P}$ in rendez-vous semantics, where the last step gives rise to $\ddot{w}$. Let the update set in the last step be $\Delta_{k^\prime}^\prime$, i.e. $S_{k^\prime+1}^\prime = S_{k^\prime}^\prime + \Delta_{k^\prime}^\prime$ with $\Delta_{k^\prime}^\prime = \bigcup_{p \in P_{k^\prime}} \Delta_{a(k^\prime,p)}$ and update sets 
\[ \Delta_{a(k^\prime,p)} = \Delta_{\tau_p}(S_{a(k^\prime,p)}^\prime) = \bigcup_{{t \in \tau_p} \atop {t \in E_p(S_{a(k^\prime,p)}^\prime)}} \Delta_t(S_{a(k^\prime,p)}^\prime) \in \boldsymbol{\Delta}_p(S_{a(k^\prime,p)}^\prime) \; . \]
Analogously, $S_{k^\prime}^\prime = S_{k^\prime-1}^\prime + \Delta_{k^\prime-1}^\prime$ with $\Delta_{k^\prime-1}^\prime = \bigcup_{p \in P_{k^\prime-1}} \Delta_{a(k^\prime-1,p)}$ and update sets 
\[ \Delta_{a(k^\prime-1,p)} = \Delta_{\tau_p^\prime}(S_{a(k^\prime-1,p)}^\prime) = \bigcup_{{t \in \tau_p^\prime} \atop {t \in E_p(S_{a(k^\prime-1,p)}^\prime)}} \Delta_t(S_{a(k^\prime-1,p)}^\prime) \in \boldsymbol{\Delta}_p(S_{a(k^\prime-1,p)}^\prime) \; . \]

If $i \stackrel{m_1}{\rightarrow} j$ appears in $\bar{w}$ and $k \stackrel{m_2}{\rightarrow} \ell$ appears in $\ddot{w}$, then there are transitions $(q_i, !m_1^{i \rightarrow j} q_i^\prime) \in \tau_i$ and $(q_k, !m_2^{k \rightarrow \ell} q_k^\prime) \in \tau_k^\prime$. As $\mathcal{P}$ is choreography-defined, we also have $(q_i, i \stackrel{m_1}{\rightarrow} j, q_i^\prime) \in \tau \in \delta$ and $(q_k, k \stackrel{m_2}{\rightarrow} \ell, q_k^\prime) \in \tau^\prime \in \delta$. We further have $(q_i, ?m_1^{i \rightarrow j}, q_i^\prime) \in \tau_j$ and $(q_k, ?m_2^{k \rightarrow \ell}, q_k^\prime) \in \tau_\ell$.

There must be at least one pair with $q_k^\prime = q_i$, i.e. $\mathcal{C}$ contains a sequence $(q_k, k \stackrel{m_2}{\rightarrow} \ell, q_i) \in \tau^\prime \in \delta$ and $(q_i, i \stackrel{m_1}{\rightarrow} j, q_i^\prime) \in \tau \in \delta$.

First consider the set $P_{k^\prime}^\prime$ of those peers $i$, for which $q_i \neq q_k^\prime$ holds for all $k$. Then the transitions $(q_i, !m_1^{i \rightarrow j} q_i^\prime) \in \tau_i$ and $(q_i, ?m_1^{i \rightarrow j} q_i^\prime) \in \tau_j$ are already enabled in $S_{k^\prime-1}^\prime$, and we can build the update set
\[ \bar{\Delta}_{k^\prime} = \bigcup_{p \in P_{k^\prime-1}} \Delta_{a(k^\prime-1,p)} \cup \bigcup_{p \in P_{k^\prime}^\prime} \Delta_{a(k^\prime,p)} \; . \]
With this we obtain a run $S_0^\prime, \dots, S_{k^\prime-1}^\prime, \bar{S}_{k^\prime}$ of $\mathcal{P}$ in rendez-vous semantics with $\bar{S}_{k^\prime} = S_{k^\prime-1}^\prime + \bar{\Delta}_{k^\prime}$, and hence $w^\prime \bar{w}^\prime \in \mathcal{L}_0(\mathcal{P})$, where $\bar{w}^\prime$ is the sequence of messages $i \stackrel{m_1}{\rightarrow} j$ with $i \in P_{k^\prime}^\prime$ (and also $j \in P_{k^\prime}^\prime$).

Let $P_{k^\prime}^{\prime\prime} = P_{k^\prime} - P_{k^\prime}^\prime$ be the set of the remaining peers. As $\mathcal{C}$ satisfies the sequence condition, we either have $(q_k, i \stackrel{m_1}{\rightarrow} j, q_i) \in \tau^\prime$ or $i \in \{ k, \ell \}$. Without loss of generality we can ignore the former case, as $k, \ell$ could be added to $P_{k^\prime}^\prime$ using the same argument as before. In the latter case both $(q_i, !m_1^{i \rightarrow j} q_i^\prime)$ and $(q_i, ?m_1^{i \rightarrow j} q_i^\prime)$ are enabled in $\bar{S}_{k^\prime}$, which defines an update set $\bar{\Delta}_{k^\prime+1} = \bigcup_{p \in P_{k^\prime}^{\prime\prime}} \Delta_{a(k^\prime,p)}$ and a run $S_0^\prime, \dots, S_{k^\prime-1}^\prime, \bar{S}_{k^\prime}, \bar{S}_{k^\prime+1}$ of $\mathcal{P}$ in rendez-vous semantics with $\bar{S}_{k^\prime+1} = \bar{S}_{k^\prime} + \bar{\Delta}_{k^\prime+1}$. Hence $w^\prime \bar{w}^\prime \bar{w}^{\prime\prime} \in \mathcal{L}_0(\mathcal{P})$ with $\bar{w}^\prime \bar{w}^{\prime\prime} = \bar{w}$.\qed

\end{proof}

\section{Extensions}\label{sec:ext}

We now discuss briefly two extensions. First we may drop the restriction that the set $P$ of peers must be finite. Another extension concerns the languages defined by P2P systems, where we may drop the restriction to finite sequences of messages.

\paragraph{Infinite P2P Systems.}

P2P systems and choreographies as defined in \cite{schewe:foiks2020} assume finite sets of peers. However, this is not necessary in general. The general behavioural theory of concurrent systems \cite{boerger:acin2016} permits also countable infinite sets of agents, but restricts the number of agents to be involved in one state transition in a concurrent run to be finite. Our semantics of P2P systems in Definition \ref{def-run} is a simplified version of this notion of concurrent run, so the extension towards infinite sets of peers will merely have to modify this definition.

For a start we allow the set $P$ of peers to be infinite. Then it also makes sense to allow the set $M$ of messages to be infinite. In a P2P system as defined in Definition \ref{def-p2p} this will lead to infinite alphabets $\Sigma_p$. Thus, we also modify Definition \ref{def-ctl-machine} such that the alphabet $\Sigma$ may be infinite, while the set $Q$ of control states should remain to be finite. Then $\delta$ can be a finite or countable infinite set of finite transition relations $\tau \subseteq Q \times \Sigma \times Q$. This reflects the intention that in one step only finitely many messages can be handled in parallel, whereas it does not do any harm to permit infinitely many choices. 

All following definitions concerning enabled transitions and update sets will remain unchanged. Finally, we only have to modify Definition \ref{def-run} such that the sets $P_i$ of peers used in a transition from state $S_i$ to $S_{i+1}$ must be finite.

These extensions of the basic definitions of control state machines, P2P systems and their p2p, queue, mailbox and rendez-vous semantics carry forward to control-state choreographies. As transition sets $\tau \in \delta$ will always be finite, it is guaranteed that also in runs of a CSC $\mathcal{C}$ only finitely many updates are applied in any state transition from $S_i$ to $S_{i+1}$, i.e. there is no need to amend Definition \ref{def-csc-run}.

With these extended definitions we can explore changes to the proofs of Propositions \ref{prop-p2p-csc}, \ref{prop-csc-p2p}, \ref{prop-lang-realisable}, \ref{prop-sequence} and \ref{prop-choice} and of Theorem \ref{thm-realisability}. In all these proofs we investigate the defining runs, so the finiteness restrictions imposed in our extended definitions will be preserved. That is, all these proofs remain valid for infinite sets of peers. We omit further details.

\paragraph{Infinite Traces.}

Following the tradition in P2P systems, choreographies and the clarifications by Finkel and Lozes all trace languages contain finite sequences of messages. This is reflected in the definitions of the trace languages $\mathcal{L}_\omega(\mathcal{P})$ with $\omega \in \{ p2p, q, m \}$ and $\mathcal{L}_0(\mathcal{P})$ at the end of Section \ref{sec:p2p}. 

However, it is also perfectly possible to consider also P2P systems that are supposed not to terminate giving rise to also infinite sequences of messages. For a run $R= S_0, S_1, \dots$ let $\mathcal{L}(R)$ be the set of all sequences with elements in $M$ that result from concatenation of $\hat{M}_0, \hat{M}_1, \dots$. Finally, define $\mathcal{L}_{\infty}$ as the language of all sequences in $\bigcup_{R} \mathcal{L}(R)$, where the union is built over all runs of the P2P systems. We write $\mathcal{L}_{0,\infty}$, $\mathcal{L}_{p2p, \infty}$, $\mathcal{L}_{q,\infty}$ and $\mathcal{L}_{m,\infty}$ for these languages in rendez-vous, p2p, queue and mailbox semantics, respectively. We call these languages the {\em infinite trace languages} of the P2P system in the different semantics. Analogously, we can define the {\em stable infinite trace languages} of the P2P system, and denote them as $\hat{\mathcal{L}}_{0,\infty}$, $\hat{\mathcal{L}}_{p2p, \infty}$, $\hat{\mathcal{L}}_{q,\infty}$ and $\hat{\mathcal{L}}_{m,\infty}$, respectively.

We can then explore how the proofs of Propositions \ref{prop-p2p-csc}, \ref{prop-csc-p2p}, \ref{prop-lang-realisable}, \ref{prop-sequence} and \ref{prop-choice} and of Theorem \ref{thm-realisability} are affected, if instead of the trace languages we exploit the infinite trace languages. Clearly, the proofs of Propositions \ref{prop-p2p-csc}, \ref{prop-csc-p2p}, \ref{prop-sequence} and \ref{prop-choice} can remain unchanged, as it suffices to look only locally at the defining runs. The argument in the proof of Proposition \ref{prop-lang-realisable} remains also valid, as stable finite states are only relevant for the finite sequences in the languages.

This leaves us with the proof of Theorem \ref{thm-realisability}, where in addition to the induction proof for finite sequences, which remains valid without any need for a change, we have to consider infinite sequences. For these every finite prefix will be in $\mathcal{L}_\omega(\mathcal{P})$ with $\omega \in \{ p2p, q, m \}$ and hence in $\mathcal{L}_0(\mathcal{P})$ due to the already conducted proof. This immediately implies that the infinite sequence will be in $\mathcal{L}_{0,\infty}(\mathcal{P})$. This shows that Theorem \ref{thm-realisability} also holds for the infinite trace languages.

\section{Conclusions}\label{sec:fin}

Choreographies prescribe the rendez-vous synchronisation of messages in P2P systems defined by communicating finite state machines. The synchronisability problem is to decide, whether a reification by asynchronous peers operating on message queues or mailboxes is equivalent to the rendez-vous composition of the peers. This problem has a long tradition marked with rather vague formalisations and wrong decidability claims \cite{basu:tcs2016}. In 2017 Finkel and Lozes finally proved undecidability in general \cite{finkel:icalp2017}. 

However, all counterexamples were somehow artificial leaving the impression that under proper restrictions it should be possible to achieve decidability. This was shown in \cite{schewe:foiks2020} for choreography-defined P2P systems. Decisive for this result was that if we start from a choreography, we only need to consider P2P systems defined by projection, and for these synchronisability coincides with language synchronisability. On these grounds it was possible to discover two necessary conditions for realisability, the sequence and the choice condition, which together are also sufficient. This closed the question of the realisability of choreographies.

Nonetheless, despite its solid theory we felt that the value of this result is limited, as the notion of choreography defined by FSMs is rather weak; it enforces a strictly sequential behaviour of the P2P system. Therefore, we asked how a more powerful notion of choreography could be defined, for which the realisability results could be preserved. In this article we introduced control state choreographies, for which this is the case. Basically, we just replaced FSMs by control-state machines, which introduce parallelism into the P2P systems. On these grounds the results from \cite{schewe:foiks2020} could be generalised to CSCs using slightly generalised definitions for the sequence and choice condition.

Still the parallelism and the choice in choreographies and P2P systems is bounded. So a natural continuation of the research would be to look at unbounded choice and unbounded parallelism. However, this requires to consider not just messages, but also other state changes (see e.g. the handling of control-state ASMs in \cite{boerger:2003}). In this context it becomes questionnable to emphasise synchronisability. In fact, it seems reasonable to focus on realisable choreographies only on high levels of abstraction, whereas in a thorough refinement process may weaker synchronisation properties are well acceptable. See e.g. the study on memory-management systems \cite{schewe:medi2019}, where this was the case.

\bibliographystyle{acm}
\bibliography{choreo}

\begin{thebibliography}{10}

\bibitem{alur:tse2003}
{\sc Alur, R., Etessami, K., and Yannakakis, M.}
\newblock Inference of message sequence charts.
\newblock {\em {IEEE} Trans. Software Eng. 29}, 7 (2003), 623--633.

\bibitem{basu:www2011}
{\sc Basu, S., and Bultan, T.}
\newblock Choreography conformance via synchronizability.
\newblock In {\em Proceedings of the 20th International Conference on World
  Wide Web ({WWW} 2011)\/} (2011), S.~Srinivasan et~al., Eds., {ACM},
  pp.~795--804.

\bibitem{basu:fase2016}
{\sc Basu, S., and Bultan, T.}
\newblock Automated choreography repair.
\newblock In {\em Fundamental Approaches to Software Engineering -- 19th
  International Conference ({FASE} 2016)\/} (2016), P.~Stevens and A.~Wasowski,
  Eds., vol.~9633 of {\em Lecture Notes in Computer Science}, Springer,
  pp.~13--30.

\bibitem{basu:tcs2016}
{\sc Basu, S., and Bultan, T.}
\newblock On deciding synchronizability for asynchronously communicating
  systems.
\newblock {\em Theor. Comput. Sci. 656\/} (2016), 60--75.

\bibitem{basu:popl2012}
{\sc Basu, S., Bultan, T., and Ouederni, M.}
\newblock Deciding choreography realizability.
\newblock In {\em Proceedings of the 39th {ACM} {SIGPLAN-SIGACT} Symposium on
  Principles of Programming Languages ({POPL} 2012)\/} (2012), J.~Field and
  M.~Hicks, Eds., {ACM}, pp.~191--202.

\bibitem{benabdallah:tacas1997}
{\sc Ben-Abdallah, H., and Leue, S.}
\newblock Syntactic detection of process divergence and non-local choice in
  message sequence charts.
\newblock In {\em Tools and Algorithms for Construction and Analysis of Systems
  ({TACAS} '97)\/} (1997), E.~Brinksma, Ed., vol.~1217 of {\em Lecture Notes in
  Computer Science}, Springer, pp.~259--274.

\bibitem{benyagoub:jsep2019}
{\sc Benyagoub, S., {A\"{\i}t Ameur}, Y., Ouederni, M., Mashkoor, A., and
  Medeghri, A.}
\newblock Formal design of scalable conversation protocols using {Event-B}:
  Validation, experiments and benchmarks.
\newblock {\em Journal of Software: Evolution and Process 32}, 2 (2020).

\bibitem{benyagoub:abz2020}
{\sc Benyagoub, S., {A\"{\i}t Ameur}, Y., and Schewe, K.-D.}
\newblock {Event-B}-supported choreography-defined communicating systems --
  correctness and completeness.
\newblock In {\em Rigorous State-Based Methods - 7th International Conference
  ({ABZ} 2020)\/} (2020), A.~Raschke, D.~M\'{e}ry, and F.~Houdek, Eds.,
  vol.~12071 of {\em Lecture Notes in Computer Science}, Springer,
  pp.~155--168.

\bibitem{benyagoub:nfm2018}
{\sc Benyagoub, S., Ouederni, M., {A\"{\i}t Ameur}, Y., and Mashkoor, A.}
\newblock Incremental construction of realizable choreographies.
\newblock In {\em {NASA} Formal Methods -- 10th International Symposium ({NFM}
  2018)\/} (2018), A.~Dutle, C.~A. Mu{\~n}oz, and A.~Narkawicz, Eds.,
  vol.~10811 of {\em Lecture Notes in Computer Science}, Springer, pp.~1--19.

\bibitem{benyagoub:medi2016}
{\sc Benyagoub, S., Ouederni, M., Singh, N.~K., and {A\"{\i}t Ameur}, Y.}
\newblock Correct-by-construction evolution of realisable conversation
  protocols.
\newblock In {\em Model and Data Engineering -- 6th International Conference
  ({MEDI} 2016)\/} (2016), L.~Bellatreche et~al., Eds., vol.~9893 of {\em
  Lecture Notes in Computer Science}, Springer, pp.~260--273.

\bibitem{bocchi:lmcs2020}
{\sc Bocchi, L., Melgratti, H.~C., and Tuosto, E.}
\newblock On resolving non-determinism in choreographies.
\newblock {\em Log. Methods Comput. Sci. 16}, 3 (2020).

\bibitem{boerger:acin2016}
{\sc B\"{o}rger, E., and Schewe, K.-D.}
\newblock Concurrent {Abstract State Machines}.
\newblock {\em Acta Inf. 53}, 5 (2016), 469--492.

\bibitem{boerger:jucs2017}
{\sc B\"{o}rger, E., and Schewe, K.-D.}
\newblock Communication in {Abstract State Machines}.
\newblock {\em J. {UCS} 23}, 2 (2017), 129--145.

\bibitem{boerger:2003}
{\sc B{\"o}rger, E., and St{\"a}rk, R.~F.}
\newblock {\em Abstract {S}tate {M}achines. A Method for High-Level System
  Design and Analysis}.
\newblock Springer, 2003.

\bibitem{brand:jacm1983}
{\sc Brand, D., and Zafiropulo, P.}
\newblock On communicating finite-state machines.
\newblock {\em J. {ACM} 30}, 2 (1983), 323--342.

\bibitem{carbone:entcs2007}
{\sc Carbone, M., Honda, K., and Yoshida, N.}
\newblock A calculus of global interaction based on session types.
\newblock {\em Electron. Notes Theor. Comput. Sci. 171}, 3 (2007), 127--151.

\bibitem{chambart:concur2008}
{\sc Chambart, P., and Schnoebelen, P.}
\newblock Mixing lossy and perfect fifo channels.
\newblock In {\em Concurrency Theory, 19th International Conference ({CONCUR}
  2008)\/} (2008), F.~van Breugel and M.~Chechik, Eds., vol.~5201 of {\em
  Lecture Notes in Computer Science}, Springer, pp.~340--355.

\bibitem{clemente:concur2014}
{\sc Clemente, L., Herbreteau, F., and Sutre, G.}
\newblock Decidable topologies for communicating automata with {FIFO} and bag
  channels.
\newblock In {\em {CONCUR} 2014 - Concurrency Theory - 25th International
  Conference, {CONCUR} 2014, Rome, Italy, September 2-5, 2014. Proceedings\/}
  (2014), P.~Baldan and D.~Gorla, Eds., vol.~8704 of {\em Lecture Notes in
  Computer Science}, Springer, pp.~281--296.

\bibitem{finkel:icalp2017}
{\sc Finkel, A., and Lozes, {\'E}.}
\newblock Synchronizability of communicating finite state machines is not
  decidable.
\newblock In {\em 44th International Colloquium on Automata, Languages, and
  Programming ({ICALP} 2017)\/} (2017), I.~Chatzigiannakis et~al., Eds.,
  vol.~80 of {\em LIPIcs}, Schloss Dagstuhl -- Leibniz-Zentrum f\"ur
  Informatik, pp.~122:1--122:14.

\bibitem{guanciale:lamp2019}
{\sc Guanciale, R., and Tuosto, E.}
\newblock Realisability of pomsets.
\newblock {\em J. Log. Algebraic Methods Program. 108\/} (2019), 69--89.

\bibitem{honda:jacm2016}
{\sc Honda, K., Yoshida, N., and Carbone, M.}
\newblock Multiparty asynchronous session types.
\newblock {\em J. {ACM} 63}, 1 (2016), 9:1--9:67.

\bibitem{schewe:medi2018}
{\sc Schewe, K.-D.}
\newblock Extensions to hybrid {Event-B} to support concurrency in
  cyber-physical systems.
\newblock In {\em Model and Data Engineering -- 8th International Conference
  ({MEDI} 2018)\/} (2018), E.~H. Abdelwahed et~al., Eds., vol.~11163 of {\em
  Lecture Notes in Computer Science}, Springer, pp.~418--433.

\bibitem{schewe:foiks2020}
{\sc Schewe, K.-D., {A\"{\i}t Ameur}, Y., and Benyagoub, S.}
\newblock Realisability of choreographies.
\newblock In {\em Foundations of Information and Knowledge Systems (FoIKS
  2020)\/} (2020), A.~Herzig and J.~Kontinen, Eds., vol.~12012 of {\em Lecture
  Notes in Computer Science}, Springer, pp.~263--280.

\bibitem{schewe:corr2020}
{\sc Schewe, K.-D., {A\"{\i}t Ameur}, Y., and Benyagoub, S.}
\newblock Realisability of control-state choreographies.
\newblock {\em CoRR abs/2009.03623\/} (2020).

\bibitem{schewe:medi2021}
{\sc Schewe, K.-D., {A\"{\i}t Ameur}, Y., and Benyagoub, S.}
\newblock Realisability of control-state choreographies.
\newblock In {\em Model and Data Engineering ({MEDI} 2021)\/} (2021), J.~C.
  Attiogb\'{e} and S.~B. Yahia, Eds., vol.~12732 of {\em Lecture Notes in
  Computer Science}, Springer, pp.~87--100.

\bibitem{schewe:medi2019}
{\sc Schewe, K.-D., Prinz, A., and B\"{o}rger, E.}
\newblock Concurrent computing with shared replicated memory.
\newblock In {\em Model and Data Engineering - 9th International Conference
  ({MEDI} 2019)\/} (2019), K.-D. Schewe and N.~K. Singh, Eds., vol.~11815 of
  {\em Lecture Notes in Computer Science}, Springer, pp.~219--234.

\bibitem{schewe:ejc2011}
{\sc Schewe, K.-D., and Wang, Q.}
\newblock Partial updates in complex-value databases.
\newblock In {\em Information and Knowledge Bases XXII}, A.~Heimb\"{u}rger
  et~al., Eds., vol.~225 of {\em Frontiers in Artificial Intelligence and
  Applications}. IOS Press, 2011, pp.~37--56.

\bibitem{zoubeyr:sttt2017}
{\sc Zoubeyr, F., {A\"{\i}t Ameur}, Y., Ouederni, M., and Tari, K.}
\newblock A correct-by-construction model for asynchronously communicating
  systems.
\newblock {\em {STTT} 19}, 4 (2017), 465--485.

\end{thebibliography}

\end{document}